\documentclass[acmsmall]{acmart}

\usepackage{xspace}
\usepackage[normalem]{ulem}
\usepackage{algorithm}
\usepackage[noend]{algpseudocode}

\DeclareMathOperator*{\argmin}{argmin}
\DeclareMathOperator*{\conflicts}{\textsc{conflicts}}
\DeclareMathOperator*{\criticality}{\textsc{criticality}}
\DeclareMathOperator*{\routelayer}{\textsc{route-seq}}

\usepackage[]{hyperref}
\usepackage{enumitem}
\usepackage{amsthm}
\usepackage{thmtools}
\usepackage{thm-restate}
\usepackage{physics} 
\usepackage[capitalise]{cleveref}

\Crefname{algorithm}{Alg.}{Algs.}

\usepackage{subcaption}
\usepackage{booktabs}
\usepackage{qcircuit}
\usepackage{tcolorbox}
\usepackage{wrapfig}
\usepackage{pifont}
\tcbuselibrary{most}
\renewcommand{\paragraph}[1]{\vspace{.05in}\noindent\textbf{\emph{#1}.~~~}}
\usepackage{tikz}
\usetikzlibrary{positioning}
\usetikzlibrary{matrix}
\usetikzlibrary{arrows}
\newcommand{\cmark}{\color{black}{\ding{51}}}
\newcommand{\xmark}{\color{red}{\ding{55}}}
\usepackage{xcolor}
\newcommand{\cnot}{\textsc{cnot}\xspace}
\newcommand{\swap}{\textsc{swap}\xspace}
\newcommand{\T}{\textsc{t}\xspace}

\newcommand{\ours}{\textsc{dascot}\xspace}
\newcommand{\autobraid}{Autobraid\xspace}

\newcommand{\layerplace}{\ours-\textsc{map}\xspace}

\newcommand{\lookaheadroute}{\textsc{\ours-route}\xspace}
\newcommand{\interactgraph}{I}

\newcommand{\oursrandmap}{\textsc{\ours-randmap}\xspace}
\newcommand{\oursrandroute}{\textsc{\ours-randroute}\xspace}

\newcommand{\ourslimited}{\textsc{limited}\xspace}
\newcommand{\temp}{\tau}
\newcommand{\finaltemp}{\tau_f}
\newcommand{\inittemp}{\tau_0}
\newcommand{\coolrate}{r}
\newcommand{\cost}{f}

\newcommand{\squaresparse}{Square Sparse\xspace}
\newcommand{\compact}{Compact\xspace}

\newcommand{\scc}{\textsc{scmr}\xspace}
\newcommand{\scmr}{\textsc{scmr}\xspace}
\newcommand{\scr}{\textsc{scr}\xspace}
\newcommand{\rspace}{R_\textsc{space}}
\newcommand{\rtime}{R_\textsc{time}}
\newcommand{\arch}{A}
\newcommand{\circuit}{C}
\newcommand{\map}{M}

\newcommand{\graph}{G}
\newcommand{\vertices}{V}
\newcommand{\edges}{E}
\newcommand{\msf}{MS}
\newcommand{\qset}{Q}

\newcommand{\sat}{\textsc{sat}\xspace}

\newcommand{\np}{\textsc{np}\xspace}
\newcommand{\psp}{\textsc{psp}\xspace}
\newcommand{\sched}{s}
\newcommand{\jobset}{J}

\newcommand{\job}{j}
\newcommand{\depcirc}{\mathit{Dep}}
\newcommand{\cyccirc}{\mathit{Cycle}}

\newcommand{\vmap}{\textsf{map}\xspace}
\newcommand{\vedge}{\textsf{path}\xspace}
\newcommand{\vexec}{\textsf{exec}\xspace}

\newcommand{\stepcount}{t}
\newcommand{\nextl}{R_\textsc{next}}

\newenvironment{mybox}[1][gray!10]{  % Change "gray!10" to "aliceblue"
    \begin{tcolorbox}[   %% Adjust the following parameters at will.
        left=0pt,
        right=0pt,
        top=0pt,
        bottom=0pt,
        colback=#1,
        colframe=#1,
        width=0.99\dimexpr\columnwidth\relax,
        boxsep=2pt,
        arc=0pt,outer arc=0pt,
    ]
}{
    \end{tcolorbox}
}
\newtheorem{definition}{Definition}
\setcopyright{cc}
\setcctype{by}
\acmDOI{10.1145/3720416}
\acmYear{2025}
\acmJournal{PACMPL}
\acmVolume{9}
\acmNumber{OOPSLA1}
\acmArticle{82}
\acmMonth{4}
\received{2024-10-11}
\received[accepted]{2025-02-18}

\begin{document}

\title{Dependency-Aware Compilation for Surface Code Quantum Architectures}

\author{Abtin Molavi}
\orcid{0009-0006-1841-9565}
\affiliation{%
  \institution{University of Wisconsin-Madison}
  \country{USA}
}
\email{amolavi@wisc.edu}

\author{Amanda Xu}
\orcid{0009-0008-2279-5816}
\affiliation{%
  \institution{University of Wisconsin-Madison}
  \country{USA}
}
\email{axu44@wisc.edu}

\author{Swamit Tannu}
\orcid{0000-0003-4479-7413}
\affiliation{%
  \institution{University of Wisconsin-Madison}
  \country{USA}
}
\email{stannu@wisc.edu}

\author{Aws Albarghouthi}
\orcid{0000-0003-4577-175X}
\affiliation{%
  \institution{University of Wisconsin-Madison}
  \country{USA}
}
\email{aws@cs.wisc.edu}

\begin{abstract}
  Practical applications of quantum computing depend on fault-tolerant devices with error correction. 
  Today, the most promising approach is a class of error-correcting codes called \emph{surface codes}.
  We study the problem of compiling quantum circuits for quantum computers implementing surface codes.
  Optimal or near-optimal compilation is critical for both efficiency and correctness. 

  The compilation problem requires (1) \emph{mapping} circuit qubits to the device qubits and (2) \emph{routing} execution paths between interacting qubits.
  We solve this problem efficiently and
  near-optimally with a novel algorithm that exploits the \emph{dependency structure} of circuit operations
  to formulate discrete optimization problems that can be approximated via \emph{simulated annealing},
  a classic and simple algorithm.
  Our extensive evaluation shows that our approach
  is powerful and flexible for compiling realistic workloads.

\end{abstract}
  

\maketitle

\section{Introduction}

Quantum computation promises to surpass classical methods in important domains, potentially unlocking breakthroughs in materials science, chemistry, machine learning, and beyond.
However, as individual physical qubits and operations are error-prone, these applications require an error-correction scheme for detecting and correcting faults.
Quantum error-correction suppresses errors with redundancy: encoding the state of a single logical qubit using several physical qubits.
Experimentalists have recently demonstrated error suppression for a single logical qubit \cite{PhysRevLett.129.030501,Acharya2022SuppressingQE,PhysRevX.11.041058} and 
small multi-qubit systems \cite{dasilva2024demonstration,Bluvstein_2023,reichardt2024demonstrationquantumcomputationerror,Google_Quantum_AI_and_Collaborators2024-fp}.

To harness the full power of the fault-tolerant quantum computers on the horizon, we need optimizing
compilers that convert circuit-level descriptions of quantum programs to error-corrected
elementary operations while preserving as much parallelism as possible. Quantum compute is a scarce 
resource, so inefficient compilation can be extremely costly.
Further, the longer the computation, the higher the probability of logical errors, which 
affect the result.

Therefore, our goal is to answer the following question:
\begin{center}
\emph{How can we compile a given circuit for a fault-tolerant device such that
execution time is minimized?}
\end{center}
We target a well-studied type of error-correction scheme called a \emph{surface code} \cite{Knill_1998, Fowler_2012,Litinski2018AGO}. A surface code
quantum device embeds logical qubits into a two-dimensional grid of physical qubits.
Two-qubit gates impose limitations on the execution of a quantum circuit by introducing contention constraints. Each two-qubit gate occupies a path on the grid and simultaneous paths cannot cross.
Gates which can theoretically be executed in parallel may be forced into sequential execution if the path of one  ``blocks'' the other, as shown in \cref{fig:blocked}.
A compiler must carefully \emph{map} qubits to grid locations and \emph{route} two-qubit gates
such that conflicts between gates are minimized and parallelism is maximized. 
We call this the \emph{surface code mapping and routing} (\scc) problem.

\begin{figure}
  \centering
  \includegraphics[scale=0.85]{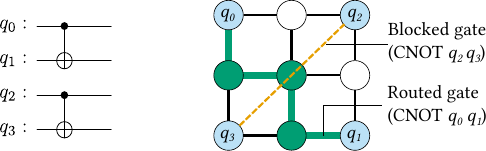}
  \caption{Spatial constraints preventing parallel execution}
  \label{fig:blocked}
\end{figure}

Existing work on the \scc problem is limited along two axes, \emph{optimality} and \emph{generality} (see \cref{tbl:comparison} for a summary):
(1) \emph{optimality}: some techniques do not optimize execution time \cite{watkins2023high}, or
optimize routing with respect to a fixed, trivial mapping \cite{PRXQuantum.3.020342}; 
(2) \emph{generality}: other techniques \cite{autobraid,Javadi-abhari-17,zhu2023ecmas} do not account for routing \T gates, a major source of two-qubit gates, and/or assume a sparse grid layout \cite{autobraid,Javadi-abhari-17,PRXQuantum.3.020342}.
This means they cannot be directly applied to other architectures which pack qubits into a smaller grid area \cite{Litinski2018latticesurgery,watkins2023high}.
Compact architectures are critical because they require fewer physical resources,
 but they make \scc more challenging. 

We present an \emph{optimizing} mapping and routing algorithm
that addresses the full \emph{generality} of the \scc problem, targeting circuits with \T gates and architectures with arbitrary layouts.
\begin{table}[t]
  \caption{Feature summary of our approach and prior work}
  \label{tbl:comparison}
  \centering
  \small
  \begin{tabular}{llll}
          \toprule
          Compiler & Optimizing? & T gates? & Any arch?\\
          \midrule
          LSC \cite{watkins2023high} & \xmark & \cmark & \cmark \\
          EDPC \cite{PRXQuantum.3.020342} & \cmark \color{black}{* (route only)} & \cmark & \xmark \\
          Autobraid  \cite{autobraid} & \cmark & \xmark & \xmark \\
          \ours (this paper) & \cmark  & \cmark & \cmark \\
          \bottomrule
  \end{tabular}

\end{table}

\begin{figure}[t]
  \includegraphics[scale=.85]{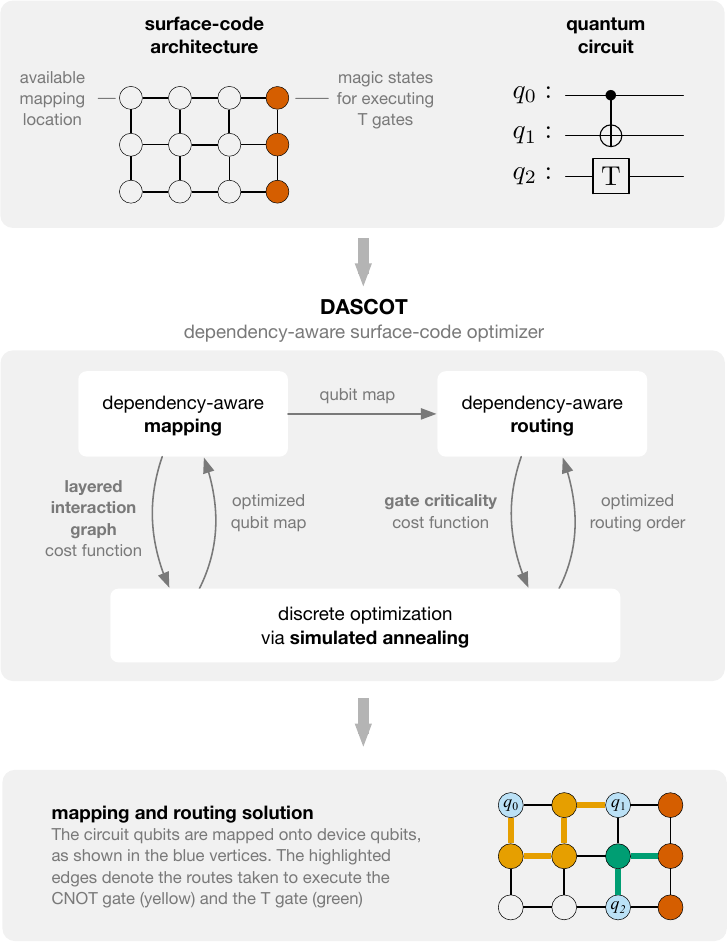}
  \caption{Overview of our approach, \ours}
  \label{fig:overview}
\end{figure}

\paragraph{The \ours approach}
Obtaining optimal or near-optimal \scc solutions is critical for both efficiency and correctness of quantum computation, but the problem is \np-complete \cite{zhu2023ecmas}.
We find that brute-force solutions based on, e.g., satisfiability solvers, do not scale to realistic problem sizes, 
as previously shown in the analogous mapping and routing problem for near-term noisy devices without error correction \cite{Wille2019MappingQC,olsq,Lin2023ScalableOL,Molavi2022QubitMA},

We present a new algorithm called \ours (Dependency-Aware Surface Code OpTimization) which mitigates the complexity of \scc by breaking mapping and routing into independent subproblems, as illustrated in \cref{fig:overview}. 
The key insight underlying \ours is \emph{dependency-awareness}: we exploit the dependency structure between circuit operations to carefully optimize the mapping and routing decisions.

\paragraph{Dependency-aware mapping}
Our goal with mapping is to map circuit qubits in such a way that is less 
likely to cause conflicts in routing gates. 
To do so, we need to track which pairs of qubits interact in the circuit.
While this idea has appeared in different guises in the qubit mapping literature, existing solutions tend to be \emph{flow-insensitive} in nature~\cite{Javadi-abhari-17,autobraid}, meaning that qubit interactions are modeled at a rough granularity that does not take into account gate ordering in the circuit.

We make the key insight that to build a good candidate mapping, we need finer-grained, flow-sensitive modeling of qubit interactions. 
For example, if one gate logically depends on another, then no valid routing solution will execute the two gates simultaneously,
so there is no reason for the mapping stage to avoid crossing paths between them. By distinguishing
between potential crossings which are likely to be problematic at the routing step and those which are not, we
can better evaluate a candidate mapping. We capture this information in a new
data structure we call a \emph{layered interaction graph}, which is the basis of our \emph{dependency-aware}
mapping approach.

\paragraph{Dependency-aware routing}
At the routing stage, our key insight is that the order in which gates are routed is crucial to the optimality of the \scc solution. Specifically, we formulate a discrete optimization problem over routing orders, with the goal 
of discovering a routing order that routes gates that are
\emph{critical} in the dependency structure of the circuit.
A critical gate has lots of other gates that consume its results later in the circuit,
so routing it early unlocks deeper layers of the circuit for routing.
Highly critical gates are prioritized to avoid a long ``tail'' of unexploited parallelism as the gates on the critical path are executed in sequence.
We observe that greedily routing gates in order of criticality, as considered in prior work \cite{Javadi-abhari-17}, does not maximize the \emph{total} criticality of routed gates because a single
gate can block the execution of multiple gates of slightly lower criticality. However, our new perspective 
on the routing problem as a search over the space of routing orders allows us to find the order which is best with respect to total criticality.

\paragraph{Reductions to discrete optimization}
The mapping and routing subproblems each induce a discrete optimization problem with respect to
a cost function. For mapping, we are searching for the mapping which minimizes the overlaps between gates in the same layer.
For routing, we seek the gate ordering which maximizes the progress towards executing the entire circuit.
In both cases, we solve these optimization problems with \emph{simulated annealing} \cite{doi:10.1126/science.220.4598.671},
a classic, easy-to-implement algorithm for solving hard, discrete optimization.
Additionally, simulated annealing supports different tradeoffs between compute time and solution quality.

\paragraph{\sat-based optimal baseline} 
We present an optimal algorithm for the \scmr problem based on a novel \sat encoding of \scmr.
Much like constraint-based approaches to mapping and routing for quantum computers without error correction \cite{olsq,Lin2023ScalableOL,shaik_et_al:LIPIcs.SAT.2024.26},
we construct a boolean formula that is satisfiable if and only if there exists an \scmr solution with a given execution time.
Then, we find an optimal solution through a sequence of calls to a \sat solver, starting from a known lower bound and incrementing until the formula is satisfiable. 
Though the \sat-based approach does not scale to large-scale applications, it allows us to find exact solutions for small circuits and to precisely 
assess the optimality of \ours for these instances. 

\paragraph{Evaluation}
We present a thorough evaluation of \ours on  
a comprehensive suite of benchmarks consisting of over 200 circuits implementing key
quantum algorithms and target two architectures representing different space--time trade-offs.
Our results confirm the efficiency and effectiveness of \ours. For example:
(1) in terms of solution quality, \ours matches or outperforms a state-of-the-art approach for \scmr without \T gates, \autobraid \cite{autobraid},
 on 84\% of benchmarks, with significant gains on \emph{dense} circuits;
  including up to 31\% cost improvement on Quantum Fourier Transform circuits; 
(2) in 69 circuits where we can compute an optimal solution using a satisfiability solver, \ours is within 25\% of optimal for all but 5;
(3) \ours can solve benchmarks with thousands of gates within minutes and can be tuned to different tradeoffs between
  compute time and optimality.

\paragraph{Contributions} Our contributions are the following:
\begin{itemize}
  \item A formalization of the surface code mapping and routing problem (\scc) capturing the combinatorial task 
  of compiling  circuits for surface code devices (\cref{sec:prob_def}).
  Compared to existing work, our formalization models the problem in more detail and with more flexibility, explicitly modeling magic states and allowing for arbitrary grids.

  \item A dependency-aware mapping algorithm using a new data structure called a \emph{layered interaction graph} (\cref{sec:mapping}).
  In comparison to prior work~\cite{Javadi-abhari-17, autobraid}, the layered nature of our interaction graphs captures finer-grained---\emph{flow-sensitive}---interaction patterns between qubits, allowing us to choose better candidate mappings that maximize parallelism.
  \item   We make the key observation is that routing order is crucial for the optimality of \scmr solutions. 
  We therefore present a dependency-aware routing algorithm which searches over routing orders and chooses the next gates to route based on progress through the dependency structure of the circuit (\cref{sec:routing}).
  This is in contrast to greedy approaches~\cite{autobraid} that choose to route gates by order of their importance, disregarding the downstream effects of such actions, and end up with suboptimal global solutions.

  \item A \sat-based \scmr algorithm for finding optimal solutions on small circuits (\cref{sec:opt-solver}).
  To our knowledge, this is the first optimal solver for the mapping and routing problem on a fault-tolerant architecture.
  The \sat-based algorithm allows us to quantify the distance from optimality for our approximate, simulated-annealing-based solutions.
  \item An empirical evaluation demonstrating the optimality and generality of \ours (\cref{sec:eval}).
\end{itemize}

\section{Surface Code Compilation: A Primer}

In this section, we describe the surface code mapping and routing problem, present illustrative examples, and provide a high-level overview of our approach.

\paragraph{Quantum error correction} Physical realizations of qubits are extremely delicate and subject to unintended changes of state, leading to computational errors.
The leading approach to address errors is called a surface code \cite{Fowler_2012}.
A surface code encodes a logical qubit using a two-dimensional lattice of physical qubits; an example of a logical qubit encoded using a surface code is shown on the left of  \cref{fig:blank_arch}. 
Each physical qubit in the lattice is designated as either a data qubit (large circles in the figure) or a measurement qubit (small circles).
Data qubits carry the state of the logical qubit, whereas measurement qubits are repeatedly measured to detect errors.
There are two types of measurement qubits (indicated by color of the surrounding region in the figure) to detect both bit-flip and phase-flip errors. As depicted on the right of \cref{fig:blank_arch}, a surface code quantum device consists of multiple surface code logical qubits encoded into one large lattice of physical qubits.

\begin{figure}[h]
  \includegraphics[width=0.46\linewidth]{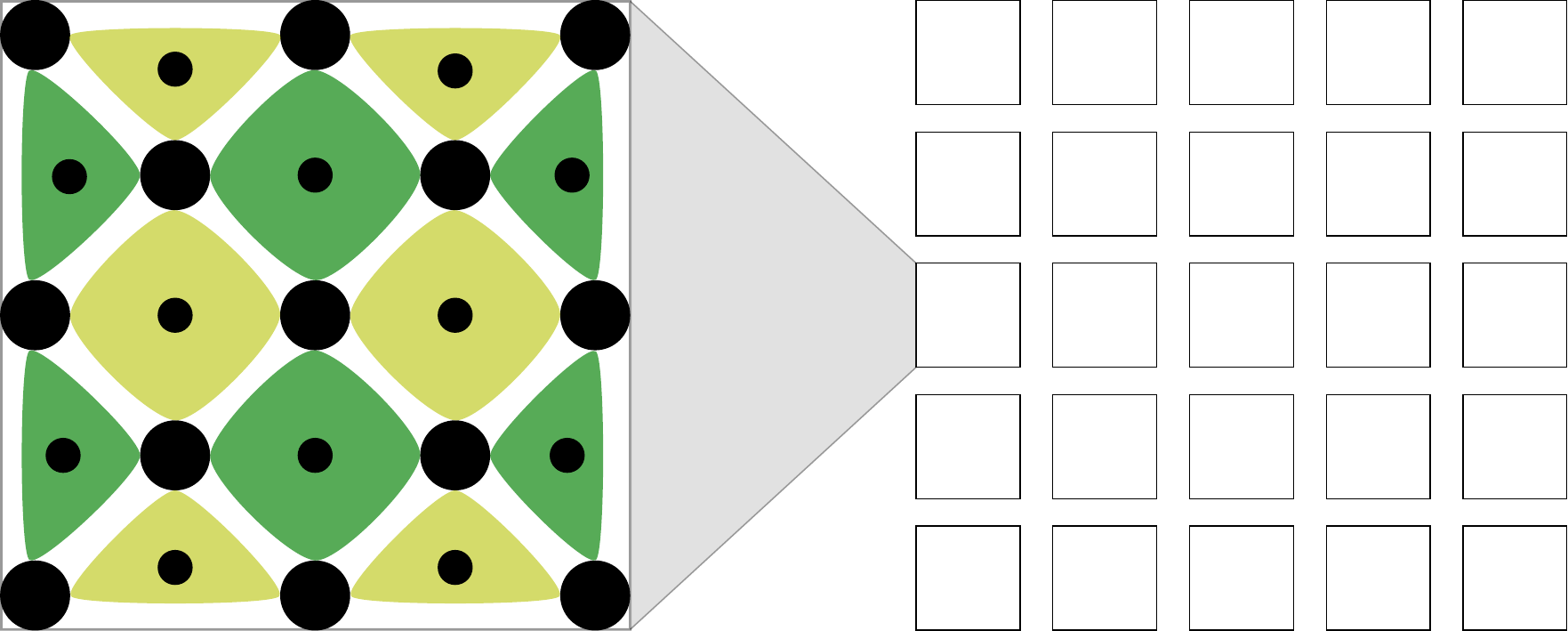}
  \caption{Quantum device implementing a surface code}
  \label{fig:blank_arch}

\end{figure}

\subsection{Surface Code Mapping and Routing by Example}
\label{sec:examples}
We will now introduce the surface code mapping and routing problem by developing some examples. The grid from \cref{fig:blank_arch}
will be our quantum computing architecture. The logical qubits in the grid may be used for storing qubits of a circuit or, as we will see shortly, for routing \cnot gates.

Suppose we wish to execute the circuit consisting of two parallel \cnot gates shown in the top left of \cref{fig:simpl_scmr}.
First, we need to choose a logical qubit to encode each of the four qubits in the circuit.
This is called a \emph{mapping}. One possible mapping is shown below the circuit. Then, we need to plan the execution of the gates. 
There are two main proposed implementations for executing \cnot gates on two logical qubits: defect braiding \cite{Fowler_2012}
and lattice surgery \cite{Horsman_2012_LS,Fowler2018LowOQ}. Both require establishing a path between the qubits. 
We focus on lattice surgery because of recent evidence that it is a more resource-efficient paradigm \cite{Fowler2018LowOQ},
but our core approach can also be applied to architectures based on the defect-braiding \cnot gate.

To apply a lattice surgery \cnot, we need to find a path of logical qubits on the grid, called ancilla qubits, from the control qubit 
to the target qubit. The path must connect a horizontal boundary of the control (the top or bottom edge) to a vertical boundary of the target,
making at least one ``bend".
The process of reserving ancillae for gates is called \emph{routing}. 
In our case, one \cnot gate requires an ancilla connection between $q_0$ and $q_1$, while the other requires one between $q_2$
and $q_3$. As long as these two connections do not overlap, the gates can be performed simultaneously
in a single \emph{time step}. A  routing solution that meets these requirements is shown in 
the bottom right of \cref{fig:simpl_scmr}.
The two routes are indicated by colored squares, with the longer green route corresponding to the \cnot gate applied to $q_0$ and $q_1$.

\begin{figure}
  \includegraphics{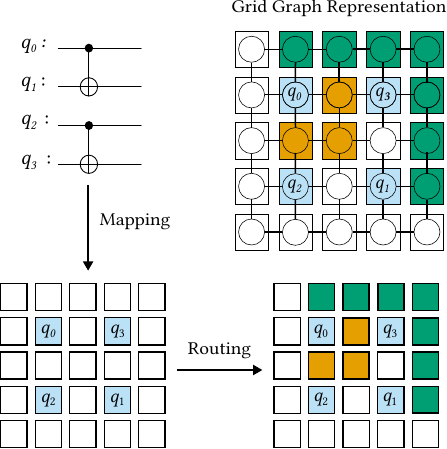}
  \caption{Simple instance of \scc}
  \label{fig:simpl_scmr}
\end{figure}

We can represent surface code mapping and routing with a grid graph
as shown in the upper right of \cref{fig:simpl_scmr}, which corresponds to the same instance and solution.
Specifically, each vertex in the grid represents a logical qubit, and we include an edge between each
adjacent pair of logical qubits, not including diagonals.

\paragraph{Optimal mapping \& routing}
The above example was fairly unconstrained. Any mapping leaving an ancilla space between each circuit qubit
has a corresponding optimal routing solution. This is not the case in general. Consider the example in \cref{fig:badmap} wherein
we are mapping the same circuit onto a $3\times3$ subset of our architecture.
The depth of this circuit is 1, so it is theoretically executable in a single time step.
However, with the choice of the mapping on the left labeled ``Mapping 1'',
there is no way to simultaneously execute the \cnot gates, as any paths from $q_0$ to $q_1$ and $q_2$ to $q_3$ will have to cross. Therefore, execution of the circuit will take two time steps.
If instead we choose ``Mapping 2'', we see that there is a routing solution that allows
for simultaneous execution, with the routes indicated by colored paths.

\begin{figure}[h]
  \includegraphics{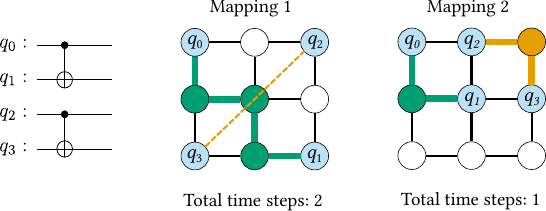}
  \caption{ Suboptimal mapping increases execution time}
  \label{fig:badmap}
\end{figure}

\paragraph{Executing T gates}
The $\T$ gate cannot be applied directly to surface code logical qubits. Instead, executing a $\T$ gate requires applying a lattice surgery \cnot between the input to the $\T$
gate and a logical qubit prepared in a so-called ``magic state'' \cite{Bravyi_2005}. We make a standard assumption
that magic state qubits are prepared in a separate region of the plane at a sufficiently high rate such that they are 
always available at dedicated storage locations \cite{PRXQuantum.3.020342}. The implication for our mapping and routing problem is that \emph{each
\T gate must be routed like a \cnot gate with the target qubit chosen from a fixed set of vertices provided as part of the input.}

We provide a simple example in \cref{fig:t-gates}. Since this circuit contains a \T gate,
we need to define magic state qubit locations. The $3\times3$ architecture has been extended with a column of magic state qubits along the right side, indicated with orange vertices. 
An optimal mapping and routing solution with one time step is shown. We have two simultaneous connections.
One is between qubits $q_0$ and $q_1$ and corresponds to the \cnot, while  
the other is between $q_2$ and
a magic state qubit, corresponding to the \T gate.
\begin{figure}
  \includegraphics{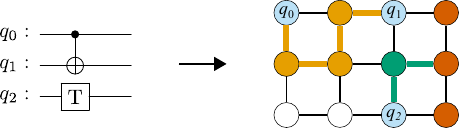}
  \caption{Example with a T gate }
  \label{fig:t-gates}
\end{figure}

\subsection{An Overview of Our Approach}
We now provide an overview of our approach (see \cref{fig:overview}).

\paragraph{Architectural flexibility}
There is a fundamental trade-off between space and time in choosing a surface code architecture.
A large grid with many available routing ancillae likely leads to shorter execution time, as many 
disjoint paths are available to route gates in parallel. A denser grid may save on qubits, but extend execution time.
The appropriate balance of qubit footprint and execution time may vary depending on characteristics of the underlying hardware and the computation
to be performed. Moreover, physical qubits are prone to defects \cite{Kreikebaum_2020,10.1145/3620665.3640362} which may force adapting the target 
architecture to omit certain regions of the grid. 

We design our approach to support arbitrary architectures. For example, given a 9 qubit circuit,
\cref{fig:arches} shows two standard architectures proposed in the literature \cite{Litinski2018AGO,Horsman_2012_LS,watkins2023high}
 which we call the \squaresparse and \compact architectures. These two architectures represent opposite ends of the space-time spectrum,
 with the \squaresparse architecture using many more qubits to allow routing more gates in parallel.

\begin{figure}[h]
  \begin{subfigure}{0.3\linewidth}
  \includegraphics[width=\linewidth]{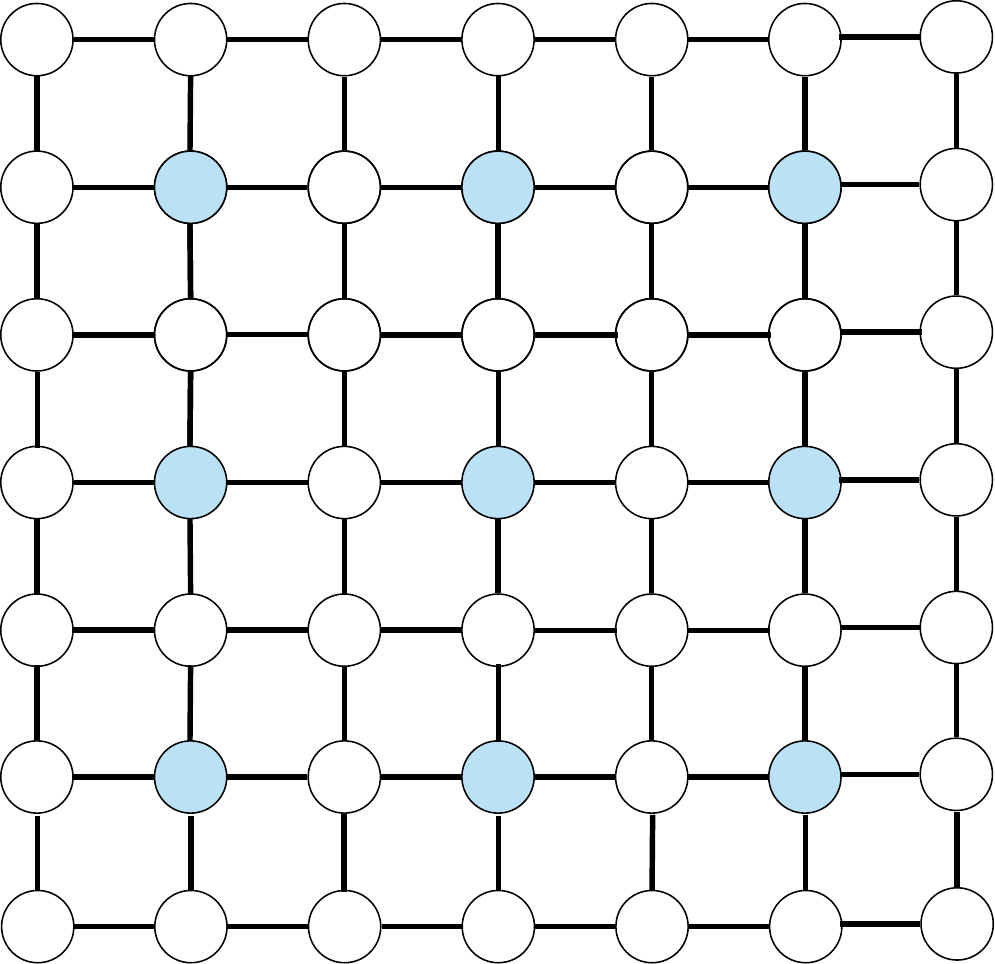}
  \caption{\squaresparse}
  \end{subfigure}
  \hspace{10pt}
  \begin{subfigure}{0.3\linewidth}
    \centering
    \includegraphics[width=\linewidth]{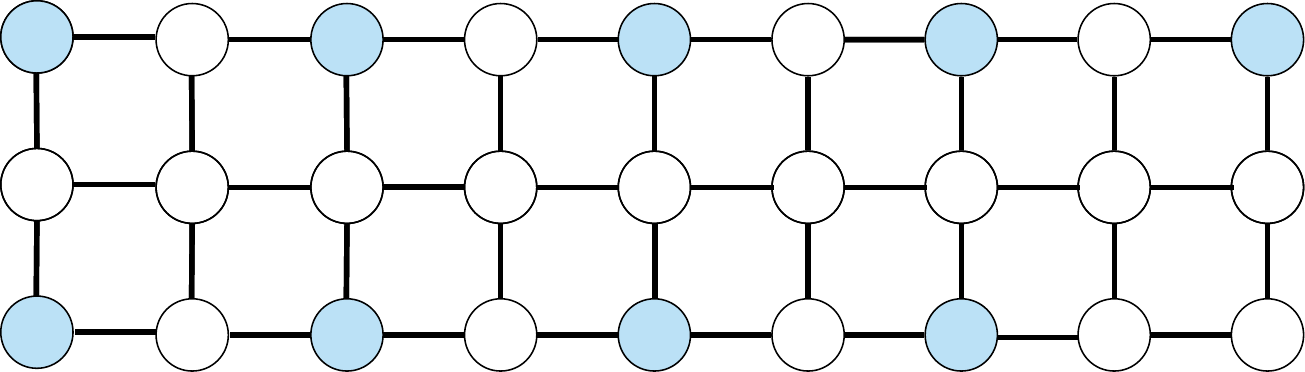}
    \caption{\compact}
    \end{subfigure}
  \caption{Architectures with different space-time tradeoffs}
  \label{fig:arches}
\end{figure}

\paragraph{Dependency-aware mapping} To determine whether a mapping is appropriate for a circuit, we need to estimate the number of times 
where routing one gate might prevent the routing of another. For example, consider the circuit in \cref{fig:map_ex}, which is a simplified version of the 4 qubit Quantum Fourier
Transform. The depth of this circuit is 5; equivalently, it can be partitioned into 5 layers of logically parallel gates (indicated with dashed lines). 

Ideally, a mapping solution should enable routing the gates of each layer simultaneously. 
Our approach efficiently searches the space of mappings for one that minimizes contention between gates within the same layer
using a simple data structure we call a \emph{layered} interaction graph. 
Here, the only nontrivial layer is layer 3.
To allow optimal routing of this layer, our mapping algorithm reserves disjoint regions (indicated with the dashed boxes) for the two gates, ensuring there exist nonoverlapping paths along which to route them.
We call this a \emph{dependency-aware} algorithm because it uses the dependency structure of the circuit 
to group gates into layers and establish which pairs of gates may be competing for routing resources. 

\begin{figure}[h]
  \includegraphics{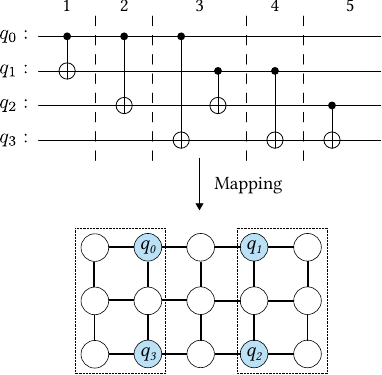}
  \caption{Illustrating dependency-aware mapping}
  \label{fig:map_ex}
\end{figure}

\paragraph{Dependency-aware routing} 
Once the qubit map is fixed, we proceed by iteratively routing as many logically parallel gates as possible.
We can reduce the space of possible routing solutions by routing each gate in sequence,
along the shortest available path. With this strategy, the order in which we iterate over the gates is important.

Consider \cref{fig:routing_ex}
with our simple circuit with one $\T$ gate and one $\cnot$. If we first route the \cnot gate along the shortest path, then the \T
gate must wait for the next time step as there is no path of available vertices between $q_2$ and a magic state qubit. On the other hand, if we 
start with the \T gate, then the \cnot can still be routed around the perimeter of the architecture, yielding a single time step.

\begin{figure}[h]
  \includegraphics[scale=1]{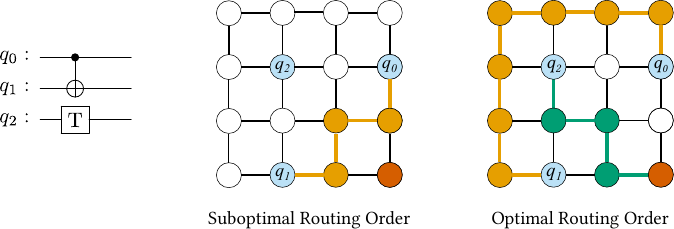}
  \caption{Illustrating dependency-aware routing}
  \label{fig:routing_ex}
\end{figure}

The key question is therefore: \emph{How do we choose the best routing order for a set of parallel gates?}
We formulate this challenge as a discrete optimization problem over routing orders. We search for an 
order which allows the simultaneous routing of many gates to make progress towards execution
of the complete circuit. Additionally, we prefer to route \emph{critical} gates, those with long chains
of dependent gates. Weighing gates by criticality prevents cascading effects from delaying a gate with many dependencies.

At the heart of dependency-aware mapping and routing is a minimization subroutine. Since our problem
is discrete in nature and infeasible to solve optimally, we use \emph{simulated annealing} to obtain solutions
which approach the minimum value. Simulated annealing is a classical, iterative optimization algorithm that allows us to trade compute time for solution quality by tuning the number of iterations.

\section{Surface Code Mapping \& Routing}
\label{sec:prob_def}
We now formalize surface code mapping and routing.

\paragraph{Architectures} Our abstraction for a fault-tolerant quantum computing architecture is a grid graph with some of the vertices
reserved for magic state qubits and others available for storing circuit qubits and routing gates. We use $\arch = (\vertices, \edges, \msf)$ to denote an architecture where $\graph = (\vertices,\edges)$ is a grid graph and  $\msf \subset \vertices$  are vertices reserved for magic state qubits.

Formally, the $m\times n$ grid graph is defined over the vertex set
$\vertices = \{1,\ldots, m\} \times \{1,\ldots, n\}$. The graph includes an edge between any pairs of vertices $(a,b)$
and $(c,d)$ such that either $|a-c| = 1$ or $|b-d|=1$. In the former case, we say that these vertices are 
\emph{horizontal neighbors}. In the latter, we say that they are \emph{vertical neighbors}.

\paragraph{Qubit maps}
Given a circuit $\circuit$ acting on a set of qubits $\qset$ and an architecture $A = (\vertices, \edges, \msf)$, a \emph{qubit map} is a function, denoted $\map$, that uniquely maps
  each qubit in $\qset$ to one of the non-magic-state qubits on the device (i.e., vertices $\vertices \setminus \msf$).

  \begin{example}
    \cref{fig:map_ex} shows a  qubit map where four qubits are mapped to four of the 15 vertices in the architecture.
    \end{example}
\paragraph{Dependency and routing}
Let $\circuit=\langle g_1, \ldots, g_k\rangle$ denote a circuit represented as a sequence of gate applications.
If two gates $g_i$ and $g_{j}$ act on a shared qubit and $i < j$, then we say
$g_j$ \emph{depends} on $g_i$. For example, in the circuit in \cref{fig:map_ex},
the gate $\cnot ~ q_1 ~ q_2$ depends on the gate $\cnot ~ q_0 ~ q_1$.
For routing, we must assign gates to execution time steps in a way that respects dependencies
in $C$. 
Also, it must be possible to execute all of the gates assigned to any particular
time step simultaneously on the architecture. Gates can be executed simultaneously if there are
vertex-disjoint paths, of the proper shape, along which to route them. We say that an
assignment of gates to time steps and paths is a \emph{valid gate route} if it meets these
requirements. In the definition to come, $\langle V \rangle$ denotes the set of finite sequences of vertices $\langle v_1, \ldots, v_m\rangle$ from a vertex set $V$.

\paragraph{Valid gate routes}
  Given an architecture $\arch = (\vertices,\edges, \msf)$,  circuit $\circuit$, and qubit map $M$,
  a \emph{valid gate route} is a natural number $t$ representing a number of time steps,
  and a pair of functions $\rspace : \circuit \to \langle V \rangle$
  and $\rtime :  \circuit \to \{1, \ldots, t\}$
  that meets the following criteria:
  \begin{itemize}
    \item \textbf{Data Preservation}: No vertex representing a circuit qubit or magic state qubit is used for routing by $\rspace$.
    \item  \textbf{\cnot Routing}: If $g =$ \cnot $q_i~ q_j$, then $\rspace(g)$ is a path from a vertical neighbor of $\map(q_i)$ to a horizontal neighbor of $\map(q_j)$.
    \item \textbf{\T Routing}: If $g = \T$ $q_i$, then $\rspace(g)$ is a path from a vertical neighbor of $\map(q_i)$ to a horizontal neighbor of a magic state.

    \item \textbf{Gate Order}: If $g_j$ depends on $g_i$,  $\rtime(g_i) < \rtime(g_j)$.
    \item \textbf{Disjoint Paths}: If $\rtime(g_i) = \rtime(g_j)$, then the paths $\rspace(g_i)$ and  $\rspace(g_j)$ share no vertices.
  \end{itemize}
\begin{example}
Consider \cref{fig:routing_ex}. Let $g_1$ be the gate  $\cnot ~ q_0 ~ q_1$ and $g_2$ be the gate  $\T~ q_2$. 
With the optimal routing order on the right, we show a valid gate route with $t=1$ time steps, so $\rtime(g_1) = \rtime(g_2) = 1$. The path 
$\rspace(g_1)$ is the longer yellow path, while $\rspace(g_2)$ is the shorter green path. 
\end{example}

With this terminology in place, the surface code mapping and routing problem (\scmr) can be concisely defined as the task of
finding a qubit map and gate route pair with the minimal number of time steps.

\begin{mybox}
  \textbf{The \scmr problem:}
  Given an architecture $\arch$ and circuit $\circuit$, find a qubit map $\map$ and a valid
  gate route $(t, \rtime, \rspace)$ such that $t$ is minimized.
\end{mybox}

\section{The \ours Algorithm}
\label{sec:relaxations}
The \scc problem is \np-complete \cite{zhu2023ecmas},  
so exact optimization is infeasible for large instances.
To simplify the problem, we decouple the mapping and routing and solve each separately.
We present novel algorithms for mapping and routing that exploit circuit dependency structure and use simulated annealing to efficiently search the space of solutions.
We start with a short primer on simulated annealing.

\subsection{Simulated Annealing Primer}
In both mapping and routing, we define a \emph{cost function} $\cost$ and search for qubit maps and routing orders (respectively) with lower cost. 
Therefore \ours must efficiently solve problems of the form
$\min_{s \in \mathcal{S}} \cost(s)$, where $\mathcal{S}$ is a large discrete set.

For this purpose, we employ \emph{simulated annealing}, a standard and simple algorithm well-suited to approximating the 
global optimum for discrete problems with large search spaces and many local minima. Simulated annealing
begins by choosing a random candidate solution $s_{\mathit{curr}}$ in $\mathcal{S}$ and an initial temperature $\inittemp$.
Then, it iteratively applies the following probabilistic local search step, repeating until the current temperature $\temp$ reaches some threshold
$\finaltemp.$
\begin{itemize}
  \item Uniformly sample a random neighbor $s_{\mathit{next}}$ of $s_{\mathit{curr}}$
  \item If $\cost( s_{\mathit{next}}) < \cost( s_{\mathit{curr}})$, set $s_{\mathit{curr}}$ to $s_{\mathit{next}}$
  \item Otherwise, do so with probability that is a function of the current temperature and cost increase
  \item Reduce the temperature by a cooling rate  $\coolrate$
\end{itemize}

The non-zero probability of transitioning to a worse solution allows escape from local minima.
We choose the standard acceptance probability 
$\exp{-\frac{\cost(s_{\mathit{new}}) - \cost(s_\textit{curr})}{\temp}}$.
Thus, transitions to a worse solution are more likely early in the optimization, when
the temperature is relatively high. 

To instantiate simulated annealing, we need to define the cost function $\cost$, 
solution space $\mathcal{S}$, and $N(s)$, the set of neighbors of a solution $s$. The algorithm also has hyperparameters:
the initial temperature $\inittemp$,  termination temperature $\finaltemp$, and  cooling rate $\coolrate$;
we describe our choices for these in \cref{sec:eval}.

\subsection{Dependency-Aware Mapping}
  \label{sec:mapping}
In the mapping stage, our goal is to find a qubit map which is likely to admit a valid gate route with
few steps without explicitly computing the full routing. 
We present an approach based on a novel data structure called a \emph{layered interaction graph}.
Whereas prior work takes a ``flattened'' view of qubit interactions with no consideration of the ordering of gates, layered interaction graphs guide our search for a mapping using the dependency structure of a circuit. 

\paragraph{Interaction graphs} A natural starting point for selecting a qubit map is to analyze which pairs of qubits interact by constructing an \emph{interaction graph}.
 \cite{Javadi-abhari-17, autobraid}. An interaction graph for a circuit $\circuit$ includes a vertex for each qubit and an undirected edge $(q_i, q_j)$ for each pair such that $\cnot ~ q_i ~ q_j$ is 
a gate in $\circuit$. To capture interactions with magic states, we also include an additional vertex $v_t$ and add an 
edge $(q_i, v_t)$ whenever there is a $\T$ gate on the qubit $q_i$. 
\begin{figure}[h]
  \includegraphics[width=0.48\linewidth]{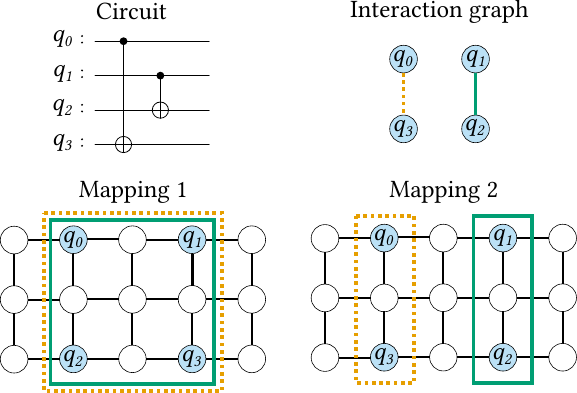}
  \caption{Interaction graphs and bounding boxes}
  \label{fig:layer-ex}
  \end{figure}
\begin{example}
\cref{fig:layer-ex} shows a circuit of 
two parallel \cnot gates. The interaction graph for this circuit (upper right) consists of two edges, one between 
$q_0$ and $q_3$ and another between $q_1$ and $q_2$.
\end{example}

\paragraph{Gate conflicts} We can 
evaluate the suitability of a map for a circuit using its interaction graph. 
For a given mapping, we can find the \emph{bounding box}
 of the interaction graph edges, following prior work \cite{autobraid}. The bounding box of an edge is defined as the smallest rectangular region on
 the architecture which contains both of the interacting qubits at its endpoints.
 We say two gates \emph{conflict} when their bounding boxes overlap because the routing
of one gate may block the routing of the other.

\begin{example}
  In the second row of \cref{fig:layer-ex}, we compare two qubit maps using the interaction graph. The bounding boxes of the two interaction graph edges with respect to each map 
  is shown in the same color and style. Following our heuristic, we prefer Mapping 2, where the gates do not conflict, to Mapping 1, where they do. 
  This is in fact the correct choice because there is a valid gate route with one step for
  Mapping 2 but not for Mapping 1.
\end{example}

\paragraph{Layered interaction graphs} However, the problem with this strategy is that it discards important
information for routing: the dependency between gates. To address this limitation, we present \emph{layered}
interaction graphs. We first partition a circuit 
into a sequence of \emph{layers} such that if a gate $g$ depends upon another $g'$, then $g$ must be assigned to a later 
layer than $g'$. Then, we build an interaction graph where each edge is labeled with the layer that the corresponding gate is assigned to.
The resulting graph with edge labels is a \emph{layered} interaction graph.
\begin{example}
A circuit based on the Quantum Fourier Transform and its layered interaction graph are
shown in \cref{fig:interactgraphs}. Notice that this circuit has 5 layers, and layer 3 is exactly the
circuit from our previous example.
\label{ex:lig}
\end{example}
\begin{figure}

  \includegraphics{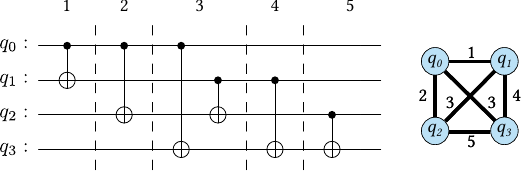}
  \caption{A circuit and its layered interaction graph}
  \label{fig:interactgraphs}
  \end{figure}

To illustrate the advantage of layered interaction graphs, in 
\cref{fig:comparing-layered}, we return to the same two
qubit maps and evaluate them with respect to the circuit and corresponding layered interaction graph from \cref{ex:lig}.

The only parallel gates in this circuit are 
the two gates in layer 3, so Mapping 2 remains a better choice because it enables routing these gates
in the same time step. Thus, an ideal cost function 
should assign lower cost to Mapping 2 than Mapping 1. Counting conflicts with respect to the standard
interaction graph instead assigns the same cost to these maps because both induce one conflict, indicated
by the yellow dashed edges. However, in the \emph{layered} interaction graph, we can distinguish between these
cases using the edge labels. If we only consider conflicts between edges with the same label, we see
one conflict under Mapping 1 in layer 3 (as in \cref{fig:layer-ex}),
and no conflicts under Mapping 2 (the edge labels 2 and 4 do not match), recovering the information that Mapping 2 is the better choice.

  \begin{figure}[h]
    \includegraphics[width=0.4\linewidth]{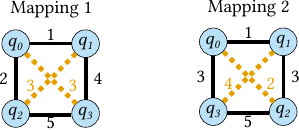}
    \caption{Qubit maps and layered interaction graph}
    \label{fig:comparing-layered}
  \end{figure}
  
  This is the key insight behind our algorithm, shown in \cref{alg:mapping}. 
  We construct a layered interaction graph (lines 2-8), then minimize \emph{conflicts} with respect to the layered interaction graph.
  Two edges are said to \emph{conflict} under a map $\map$ if 
  their bounding boxes overlap and they have the same layer label. The bounding box of an edge associated with a gate $\cnot ~q_i ~q_j$ is the bounding box of the 
  architecture vertices $\map(q_i)$ and $\map(q_j)$. The bounding box of an edge associated with a gate $\T ~ q_i$
  is the bounding box of the architecture vertices $\map(q_i)$ and the nearest magic state vertex.
  
  \begin{algorithm}
  \caption{Dependency-Aware Mapping}\label{alg:mapping}
  \begin{algorithmic}[1]
  \Procedure{\layerplace}{arch $\arch$, circuit $\circuit$, qubits $Q$}
    \State Partition $\circuit$ into layers $\mathcal{L} = L_1, \ldots, L_d$.
    \State Let $\interactgraph$ be the empty graph over the vertices $Q \cup \{v_t\}$
    \For{layer $L_k$ in $\mathcal{L}$ } 
    \For{each $\cnot ~ q_i ~q_j$  in $L_k$}
        \State Add edge $(q_i, q_j)$ to $\interactgraph$ with label $k$
    \EndFor
    \For{each $\T ~ q_i$  in $L_k$}
    \State Add edge $(q_i, v_t)$ to $\interactgraph$ with label $k$
    \EndFor
    \EndFor
   \State \Return \color{ACMDarkBlue}$\argmin_{\map} \conflicts(\interactgraph, \map)$
  \EndProcedure
  \end{algorithmic}
  \end{algorithm}

\paragraph{Solving the minimization objective}
The minimization problem (blue in \cref{alg:mapping}) can be solved via simulated annealing:
(1)  the search space $\mathcal{S}$ is the set of possible qubit maps,
(2) the neighbor set of $N(\map)$ of a map $\map$ is those obtainable 
    by swapping the locations of one pair of circuit qubits,
    and
(3) the cost function $\cost$ is the number of conflicts between edges with the same label in 
    the interaction graph.
  
\subsection{Dependency-Aware Routing}
\label{sec:routing} 
 At the routing stage, we take a circuit, architecture, and qubit map as input and return a sequence
of time steps. As formalized in the definition of a valid gate route in \cref{sec:prob_def}, 
each time step $k$ consists of (1) a set of parallel gates, corresponding to the gates $g$ for which we set $\rtime(g) = k$, and 
(2) disjoint paths along which to route them, defining the function $\rspace : \circuit \to \langle V  \rangle$ on  the same set of gates. 
The goal is to find the sequence with the fewest total time steps. The distinguishing feature of our
approach is that we define and solve a discrete optimization problem over the space of candidates for the
each time step, focusing on gates with high \emph{criticality}.

We begin by describing our algorithm for constructing a single time step
for a given layer of parallel gates, which we will iteratively apply to generate a full routing solution. 
To limit the combinatorial explosion in possible routes, we consider gates sequentially, routing each along the shortest available path, if one exists 
(otherwise the gate is delayed to a future time step). \cref{alg:one-step} shows this
subroutine for a particular sequence $\sigma$ of gates. 

With this linear approach, the order in which 
we iterate over gates affects the resulting time step because a gate earlier in the 
order may occupy vertices which prevent the routing of a later gate. Therefore, our goal is to find 
the routing order $\sigma$ that results in the best time step.

\begin{algorithm}

  \caption{Single Step Routing}
  \label{alg:one-step}
  \begin{algorithmic}[1]
    \Procedure{$\routelayer$}{gate sequence $\sigma$, arch $\arch$, $\map$}
    \State initialize $\rspace'$ to $\emptyset$
    \For{gate $g$ in $\sigma$}
      \State Let $p$ be a shortest path for routing $g$ in $\arch$  
      \State $\rspace' \gets \rspace' \cup \{(g,p)\}$ \Comment{\emph{if some  path exists}}
      \State Remove all vertices in $p$ from $A$
  \EndFor
   \State \Return $\rspace'$
  \EndProcedure
  \end{algorithmic}

\end{algorithm}

\paragraph{Shortest-first order} One approach to a routing order is to choose the next gate to be the one with the 
shortest available routing path. Intuitively, this order is often effective because shortest paths occupy fewer vertices and are thus less likely to block other paths. 
 In fact, the shortest-first order has a guaranteed approximation ratio
 derived from the literature on the \emph{node-disjoint paths} problem \cite{Kolliopoulos1997ApproximatingDP,PRXQuantum.3.020342}. 
 It results in routing $O(\textsc{opt}/\sqrt{n})$ gates where $\textsc{opt}$ is the optimal number of gates
 which can be routed in the next time step and $n$ is the number of vertices in the architecture.

\paragraph{Routing order search} While ordering the gates by path length is a sound principle, there are numerous
cases in practice where it produces suboptimal results.
\begin{example}
  Consider \cref{fig:routing_ex}. The shortest paths
available to execute the two gates are equal. With shortest-first sorting, these two
orders are equivalent. However, only one order allows routing both gates in 
the same time step. 
\end{example}
To prevent cases like this, where a fixed heuristic cannot identify an optimal
routing order, we search for the best routing order at each iteration, generating several candidate 
time steps and picking the best. To define the ``best" choice among candidate time steps, we need to estimate the degree to which a time step will lead to a high-quality overall solution, with few total time steps.

\paragraph{Dependency-aware cost}
We use a dependency-aware metric, where the cost reduction associated with routing a gate $g$ is equal to its \emph{criticality}---we want to prioritize routing important gates first.

\begin{mybox}
\textbf{Criticality:} Let $\circuit$ be a quantum circuit and $g$ be a gate in $\circuit$. 
   The \emph{criticality} of $g$, denoted $\criticality(g)$ is the depth of the subcircuit of $C$ consisting of gates which depend on $g$, including $g$ itself. 
\end{mybox}
For example, in the circuit in \cref{fig:t-gates}, both gates have a criticality of 1, while the gate $\cnot ~ q_0 ~ q_3$ in \cref{fig:map_ex} has a criticality of 3.
Intuitively, routing critical gates unlocks deeper layers of the circuit for routing.
Gates with high criticality should be prioritized to avoid a long ``tail'' of underutilized time steps as the gates on the critical
path are executed in sequence. 

Thus, we search for the routing order such that the resulting time step $s = \routelayer(\sigma, \arch, \map)$ is minimal with respect to the 
dependency-aware cost function below (highlighted in blue in \cref{alg:laroute})
$$ f(s) = \sum_{(g,p) \in s} -\criticality(g)$$

Note that this cost cannot be minimized greedily. The critical-first order, routing gates in order of criticality, is not necessarily
optimal.

\begin{example}
  \cref{fig:crit-first} depicts the front layer of a circuit with three parallel gates. Suppose the 
  \cnot gate has a criticality of 3 and the two \T gates have criticality 2. 
By first routing the most critical gate, the \cnot, we block the routing of the two \T gates. Therefore
the critical-first routing order has a cost of -3. On the other hand, if we first consider the
\T gates, both can be routed in the same time step, yielding a cost of -4. Thus, the greedy approach is not optimal here.

\end{example}
\begin{figure}
  \includegraphics[scale=1]{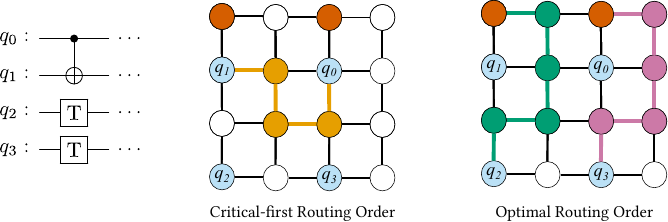}
  \caption{A case where the critical-first order is suboptimal}
  \label{fig:crit-first}
\end{figure}

Our routing algorithm is \cref{alg:laroute}. It determines the 
front layer of the circuit, applies simulated annealing to find the best choice for the next time step,
and repeats until all gates have been routed, returning a sequence of time steps. 
We instantiate simulated annealing as follows:
(1) the search space $\mathcal{S}$ is the set of orderings of the layer,
(2) the neighbor set of $N(\sigma)$ of an ordering $\sigma$ is those obtainable 
  by swapping the position of one pair of gates,
 and (3) the cost function is given by $\cost(\routelayer(\sigma, \arch, \map)).$ 

\begin{algorithm}[t] 

  \caption{Dependency-Aware Routing}\label{alg:laroute}
  \begin{algorithmic}[1]
  \Procedure{\lookaheadroute }{arch $\arch$, circuit $\circuit$, map $\map$}
    \State initialize $\rtime$ and $\rspace$ to $\emptyset$
    \State initialize number of steps $\stepcount$ to  $0$
    \While{$\circuit \neq \emptyset$} \Comment \emph{note we remove gates from $\circuit$}
        \State Let $L_1$ be the front layer of gates in $\circuit$
        \State Let $\Sigma$ be the set of all possible orderings of $L_1$
        \State $\sigma^* \gets$ {\color{ACMDarkBlue}  $\argmin_{\sigma \in \Sigma} \cost(\routelayer(\sigma, \arch, \map))$ 
        \Comment{\emph{simulated annealing}}}
        \State $\nextl\gets \routelayer(\sigma^*, \arch, \map)$
        \State $\stepcount \gets  \stepcount + 1$
        \For{each (gate, path) pair $(g,p)$ in $\nextl$}
          \State $\rtime(\mathit{g}) \gets \stepcount$
          \State $\rspace(\mathit{g}) \gets \mathit{p}$ 
          \State Remove $\mathit{g}$ from $\circuit$ 
        \EndFor
    \EndWhile
   \State \Return $\rtime, \rspace$
  \EndProcedure
  \end{algorithmic}
  \end{algorithm}

\section{A \sat-Based Optimal Baseline}
\label{sec:opt-solver}
To evaluate the quality of solutions produced by \ours, we developed a \sat-based optimal solver for the \scc problem.
Though the \sat approach does not scale to large circuits, it provides a ``ground truth" for small circuits.
Like approaches for \textsc{nisq} computers~\cite{olsq,Lin2023ScalableOL,shaik_et_al:LIPIcs.SAT.2024.26},
the optimal solver encodes the existence of a solution with $k$ time steps as a query to a boolean satisfiability solver, incrementing
$k$ if the result is ``unsatisfiable'', and iterating until a solution is found. In the rest of this section, 
we describe the
encoding for a fixed value of $k$.

Throughout we will fix
an architecture $\arch = (\vertices,\edges, \msf)$,  a circuit $C=\langle g_1, \ldots, g_k\rangle$ acting on qubits in a set $Q$ and a time limit $t_s$.
Our encoding includes three types of boolean variables: the \vmap variables represent the qubit mapping, the
\vexec variables represent the assignment of gates to time steps (the function $\rtime$) and the
\vedge variables represent the path used to route gates (the function $\rspace$). The formula over
these variables that we produce will enforce that the qubit map encoded by
\vmap and the routing encoded by \vexec and \vedge are a valid solution as defined in \cref{sec:prob_def}.
In describing these constraints, we use cardinality constraints which enforce that \emph{at most one} (\textsc{amo}) or
\emph{exactly one} (\textsc{eo}) of a set of variables is set to true. The efficient translation of cardinality constraints
into conjunctive normal form is a well-studied area \cite{Bailleux2003EfficientCE, ansotegui-sat04,Sinz2005TowardsAO,gent-modref04,prestwich-encodings} so we leave the encoding unspecified for simplicity.

\paragraph{Variables}
The boolean variables that appear in our encoding are listed below along with their intended semantics.
\begin{itemize}
  \item $\vmap(q, v)$: the qubit $q$ in the circuit is mapped to the vertex $v$.
  \item  $\vexec(g, t)$:  the gate $g$ is
  executed in time step $t$.
  \item $\vedge(u,v, g, t)$: the edge $(u,v)$ is used
  in the path for gate $g$ at time step $t$.
\end{itemize}

Next, we define the constraints over these variables that ensure a satisfying 
assignment corresponds to a mapping and routing solution.

\paragraph{Maps are injective functions} To enforce that a satisfying assignment corresponds to a valid qubit map,
we require that for each $q$ there is \emph{exactly one} $v$ where $\vmap(q, v)$ is set to true (functional consistency)
and that for each $v$ there is \emph{at most one} $q$ where $\vmap(q, v)$ is set to true (injectivity).
These constraints are shown below.

\begin{align*}
  \\ \textsc{map-valid} &\triangleq \bigwedge_{q \in \qset} \textsc{eo}(\{\vmap(q,v) : v \in \vertices \setminus \msf \} ) \land \bigwedge_{v \in \vertices \setminus \msf} \textsc{amo}(\{\vmap(q,v) : q \in \qset\} )
\end{align*}

\paragraph{Paths preserve data} 
We cannot overwrite circuit qubits or magic state qubits in the routing process. In terms of \vedge and \vmap
variables, this means that we cannot set both $\vedge(u,v, g, t)$ and $\vedge(v,w, g, t)$ for any $u$, $w$, $g$, and $t$ if 
we have set $\vmap(q, v)$ for some qubit $q$ or if $v \in MS$.

\begin{align*}
  \textsc{circ-safe} &\triangleq \bigwedge_{\substack{ u, v, w \in \vertices \\ t \leq t_s \\ q\in \qset, g \in \circuit}} \ \lnot \vmap(q, v) \lor \lnot \vedge(u,v, g, t)  \lor  \lnot \vedge(v,w, g, t) \\
  \textsc{ms-safe}   &\triangleq \bigwedge_{\substack{v \in \msf, u, w \in \vertices \\ t \leq t_s \\ g \in \circuit}} \lnot \vedge(u,v, g, t) \lor \lnot \vedge(v,w, g, t) \\
  \textsc{data-safe}   &\triangleq \textsc{circ-safe} \land \textsc{ms-safe}
\end{align*}

\paragraph{All gates are executed in logical order} To ensure that all gates are executed in a valid order, we enforce that for
each $g$, there is exactly one $t$ such that $\vexec(g, t)$ is set to true and that if $g_j$ depends on $g_i$, (which we denote with $g_i \prec g_j$ ),
then $\vexec(g_i,t)$ implies $\vexec(g_j, t')$ for some $t' > t$. This is represented by the constraints:
\begin{align*}
  \textsc{gates-ordered} & \triangleq \bigwedge_{g \in \circuit}\textsc{eo}(\{\vexec(g, t) : t \leq t_s\}) \land \bigwedge_{\substack{g \prec g' \\t' \leq t}} \lnot\vexec(g,t) \lor \lnot\vexec(g', t')
\end{align*}

\paragraph{Paths are disjoint}  Though the edges in the underlying graph are undirected,
it is useful for our formulation to assign a direction to the paths, so we assign the direction of the edge as
from $u$ to $v$ if $\vedge(u,v, g, t)$ is set to true. We require that our simultaneous paths are \emph{vertex}
disjoint, meaning each vertex can have at most one incoming edge and at most one outgoing edge. Therefore, for each time step $t$ and vertex $u$,
we include the constraint that there is
at most one pair $v, g$ such that $\vedge(u,v, g, t)$ is set to true and at most one pair $v', g'$
such that $\vedge(v',u, g', t)$ is set to true:
\begin{align*}
  \textsc{disjoint} \triangleq \bigwedge_{\substack{t \leq t_s \\ u \in V}} &\textsc{amo}(\{\vedge(u, v, g, t) : v \in V, g \in \circuit \})
  \land \textsc{amo}(\{\vedge(v', u, g', t) : v' \in \vertices, g' \in \circuit \})
\end{align*}

\paragraph{Paths connect \cnot pairs} The most complicated type of constraints are those that enforce
that the \vedge variables encode a valid path between the relevant vertices. We start with the case of
\cnot gates. The constraints effectively construct a path inductively. Let $g = \cnot$ $q_i~ q_j$ be a \cnot gate.
The path corresponding to $g$ must begin at a vertical neighbor of the control qubit $q_i$ and terminate at a horizontal neighbor of the target qubit $q_j$.
This leads to the ``base case'' constraint for each end of the path. We state them below 
using $VN(v)$ (and $HN(v)$) as an abbreviation for the set of up to two vertical (and horizontal) neighbors of a vertex $v$:
\begin{align*}
  \textsc{cnot-start} &\triangleq \bigwedge_{\substack{v\in \vertices \\ g(q_i, q_j) \in \circuit \\ t \leq t_s}}\left( \lnot \vmap(q_i, v) \lor \lnot \vexec(g, t) \lor \bigvee_{u \in VN(v)}\vedge(v, u, g, t)\right) \\
  \textsc{cnot-reach-target} &\triangleq \bigwedge_{\substack{v\in \vertices \\ g(q_i, q_j) \in \circuit \\ t \leq t_s}}\left( \lnot \vmap(q_j, v) \lor \lnot \vexec(g, t) \lor \bigvee_{u \in HN(v)}\vedge(u, v, g, t) \right)
\end{align*}

The inductive constraint is that for each edge $(u,v)$ on the path, either $u$ is a valid starting point for the path,
or an internal vertex on the path. In the latter case, it must be incident to another edge to follow back towards the starting point. In our
variables, we can express the two options as $\vedge(u, v, g, t)$ implies either $\vedge(w, u, g, t)$ for some distinct vertex $w$ adjacent to $u$ or $\vmap(q_i, u)$.

\begin{align*}
  \textsc{cnot-inductive}  &\triangleq \bigwedge_{\substack{g(q_i, q_j) \in \circuit       \\ t \leq t_s \\ (u,v)\in E}}\left( \lnot \vedge(u, v, g, t) \lor \bigvee_{\substack{(w,u) \in \edges \\ w \neq v}} \vedge(w, u, g, t) \lor \vmap(q_i, u) \right) \\
  \textsc{cnot-routed} &\triangleq \textsc{cnot-start} \land \textsc{cnot-reach-target} \land \textsc{cnot-inductive}
\end{align*}

\paragraph{Paths connect \T targets to magic states} Let $g=$ \T $q_j$ be a \T gate. We also need
to enforce that the \vedge variables assign an appropriate path between $q_j$ and a magic state qubit. This type of constraint is similar to the previous. 
The only difference is in the definition of valid end points which are now horizontal neighbors of magic state qubits.
\begin{align*}
  \textsc{t-reach-target}      & \triangleq \bigwedge_{\substack{g(q) \in \circuit \\ t \leq t_s}}\left( \textsc{eo}(\{\vedge(u, v, g, t) : v \in \msf, u \in HN(v)\}) \right)\\
  \textsc{t-routed} &\triangleq \textsc{t-start} \land \textsc{t-reach-target} \land  \textsc{t-inductive}
\end{align*}

The formula corresponding to an \scc instance is given by the conjunction of these constraints:

\begin{align*}
  \varphi(\arch, \circuit, t_s) \triangleq \textsc{map-valid} \land \textsc{gates-ordered} \land \textsc{data-safe} \land \textsc{disjoint} \land \textsc{cnot-routed} \land \textsc{t-routed} 
\end{align*}

\begin{theorem}
  The formula $\varphi(\arch, \circuit, t_s)$ is satisfiable if and only if 
  there is a map $\map$ and corresponding valid gate route $(t_s, \rtime, \rspace)$ for the \scc problem given by $\arch$ and $\circuit$, 
  and there is an explicit translation from \scc solutions $\map, (t_s, \rtime, \rspace)$ to models of $\varphi(\arch, \circuit, t_s)$ and vice-versa.
  \label{thm:enc-correct}
\end{theorem}

\section{Implementation and Evaluation} 
\label{sec:eval}
We now present our evaluation of \ours.

\paragraph{Benchmarks} Our evaluation is performed on a benchmark suite of 232 application circuits. Our suite includes the  entire set collected by \citet{zulehner2018efficient}. This existing set consists of circuits  derived from the RevLib suite \cite{wille2008revlib} and programs written in the Quipper \cite{green2013quipper} and ScaffoldCC \cite{JavadiAbhari2014} quantum programming languages.   We extended 
it with implementations of major quantum algorithms: Shor's Algorithm \cite{shor1994}, the Quantum Fourier Transform \cite{coppersmith2002qft}, Bernstein-Vazirani \cite{qcc}, QAOA \cite{farhi2014qaoa}, and Grover's Algorithm \cite{grover1996}.

\paragraph{Architectures}
We target the two architectures \squaresparse and \compact from
\cref{fig:arches}.
The \squaresparse architecture represents a case with a large space footprint studied in prior work on mapping and routing \cite{autobraid,watkins2023high}. %
It is the smallest square grid which includes enough space to surround each circuit qubit with routing ancillae on 
all sides. For a circuit with $n$ qubits, this results in a side length of $2\lceil \sqrt{n} \rceil + 1$. The \compact architecture,
in contrast, is designed to minimize the required device qubits.
This architecture is nearly linear, with three rows, and enough columns to include an ancilla qubit between each circuit qubit ($n-1$ columns for even $n$). 
 In both cases, we assume this mapping region is surrounded on all sides by magic state qubits.

 \paragraph{Experimental setup}
All runs were allotted a 1hr timeout on one core of an AMD EPYC\texttrademark ~ 7763 2.45 GHz Processor and 32GB of RAM accessed via a distributed research cluster.
Since \ours is a randomized algorithm, we collected 20 trials and plot the mean solution and a 95\% confidence interval.
We choose the simulated annealing parameters $\inittemp = 100$, $\finaltemp = 0.1/d$, $\coolrate = 0.1/d$  (where $d$ is the depth of the circuit)
for mapping and $\inittemp = 10$, $\finaltemp = 0.1$, $\coolrate = 0.1$ for routing. The constants were determined empirically 
with a grid search over the ranges $[1,10^3]$ for $\inittemp$ and $[10^{-3}, 1]$ for $\coolrate$ and $\finaltemp$,  
using randomly generated circuits to avoid benchmark overfitting.

\paragraph{Research questions} Our experimental evaluation is designed to answer the following questions:

\textbf{RQ1:}  How does our approach compare to prior work on surface code mapping and routing?

 \textbf{RQ2:} What are the gains from our  mapping and routing algorithms over greedy baselines?

 \textbf{RQ3:} How close to optimal are our solutions?
 
 \textbf{RQ4:} What is the tradeoff between optimality and scalability?

\subsection{(RQ1) Comparison to Prior Work}
 To answer RQ1, we primarily compare \ours to \autobraid \cite{autobraid}, a state-of-the-art algorithm
for a closely related problem (though other tools are represented in addressing RQ2). \autobraid targets the defect-braiding \cnot gate and assumes a  ``sparse'' architecture, like \squaresparse, where logical qubits are surrounded on all
sides by routing ancillae. For a direct comparison, we implemented a version of \autobraid with two minor modifications:
(i) constraints which require paths to include a ``bend'' as needed for lattice surgery and (ii) support for more general architectures.
Moreover, \emph{\autobraid does not handle routing of \T gates.} Therefore, for the purpose of this 
comparison only, we do not include them as part of the problem instance, as a special case of our \scc problem.
This matches the experimental setup originally used to evaluate Autobraid \cite{autobraid}.

\paragraph{Overall results}
We begin by analyzing the quality of solutions produced by Autobraid and \ours across the entire benchmark suite.
In \cref{fig:summary_bars}, we count the number of benchmarks where one algorithm outperforms the other 
in terms of solution quality by producing a \emph{mean} solution with fewer time steps. The ``Match'' category 
represents circuits for which the two algorithms reached solutions with the same number of time steps.

On the \squaresparse architecture, the largest category by far is ``Match'' where the two algorithms 
reach equivalent solutions, containing  83\% of the benchmark suite. In comparison, the ``\ours better'' category contains 16\%  of circuits and ``\autobraid better''
category contains 1\% (just three circuits). Most circuits are in the ``Match'' category because, with the plentiful routing resources of this architecture, 
both algorithms often achieve the theoretical lower-bound given by the depth of the circuit. On the other hand, the 
constraints of the \compact architecture do reveal a separation between the two algorithms. \ours
outperforms or matches \autobraid on 84\% of the benchmarks and strictly outperforms \autobraid on 59\% of the benchmarks.

\begin{figure}
  \begin{subfigure}{0.285\linewidth}
  \includegraphics[width=\linewidth]{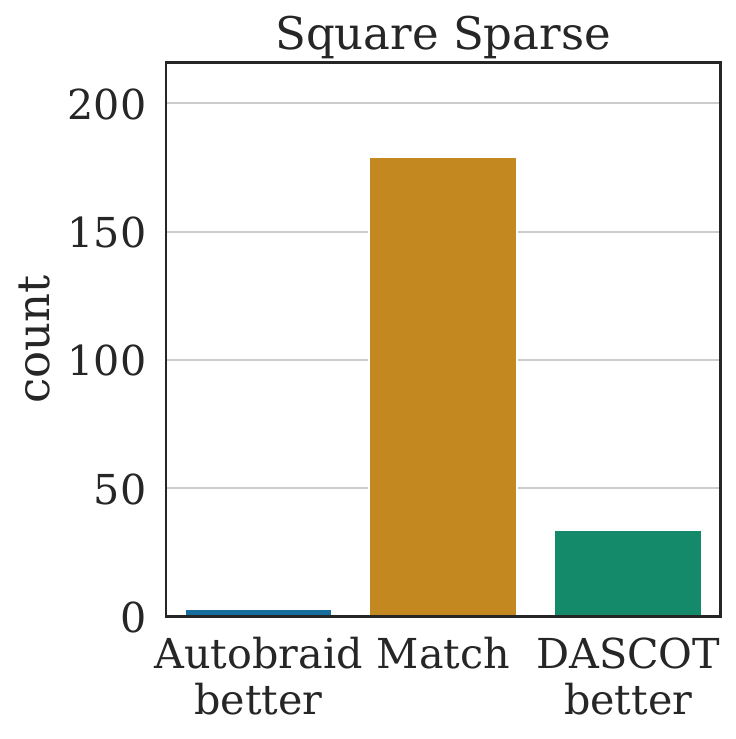}
  \end{subfigure}
  \begin{subfigure}{0.285\linewidth}
    \includegraphics[width=\linewidth]{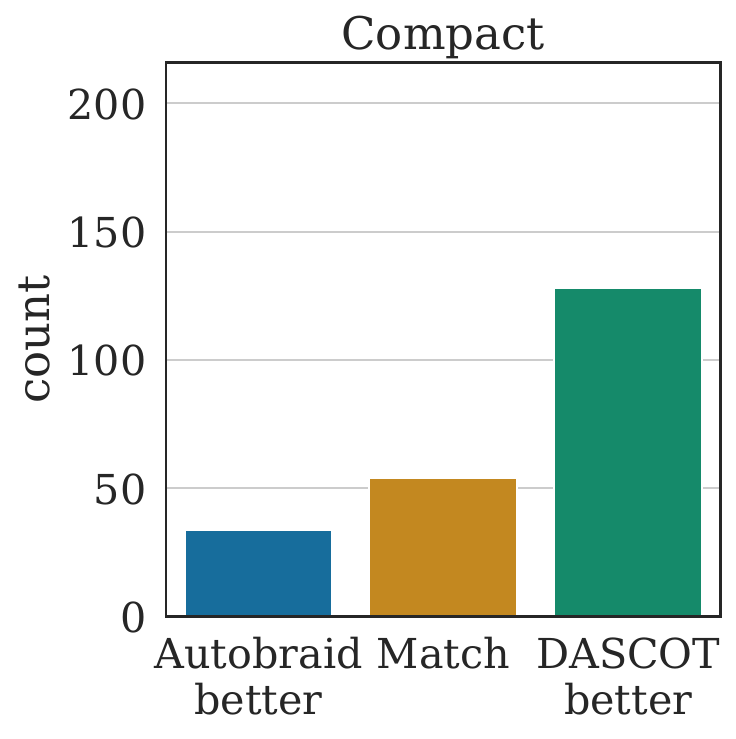}
    \end{subfigure}
  \caption{Summary Comparison with Autobraid}
  \label{fig:summary_bars}
\end{figure}

\paragraph{QFT circuits}
Quantum Fourier Transform circuits are an important application which
are particularly challenging to map and route because of their dense interaction graphs with many gates 
per layer.
However, \ours is well-suited to the complexity of QFT circuits. \cref{fig:qft} compares the \emph{cost ratios} of
the solutions provided by \ours and \autobraid for these hard instances.   
The cost ratio is defined as the quantity:
$$
\frac{\text{\# of time steps in solution}}{\text{circuit depth}}
$$
Since the depth of the circuit is a theoretical lower bound on the number of time steps in a solution,
the best cost ratio is 1. 
\ours achieves lower mean cost ratios on than \autobraid  up to about a 31\% relative improvement on both architectures
for the 100 qubit circuit.

\begin{figure}[h]
  \begin{subfigure}{0.285\linewidth}
  \includegraphics[width=\linewidth]{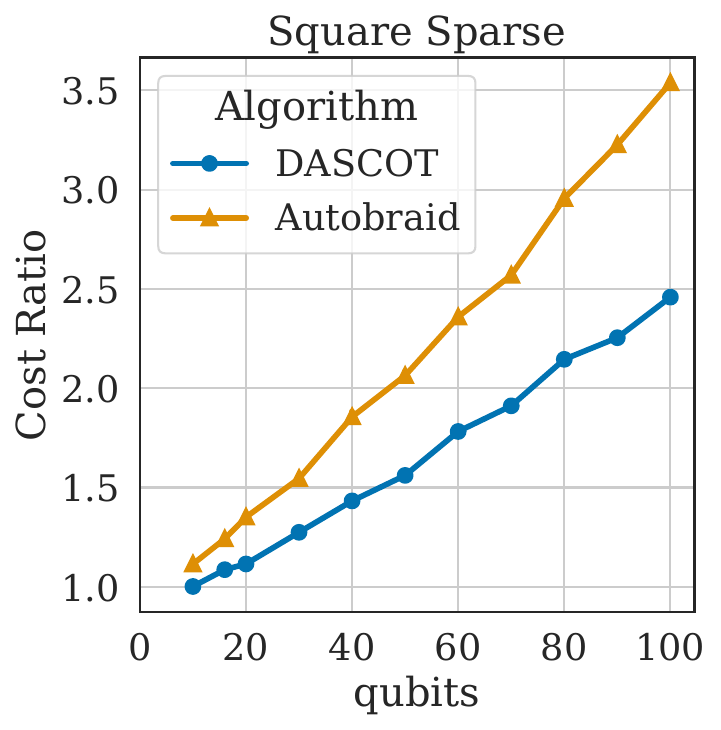}
  \end{subfigure}
  \begin{subfigure}{0.285\linewidth}
    \includegraphics[width=\linewidth]{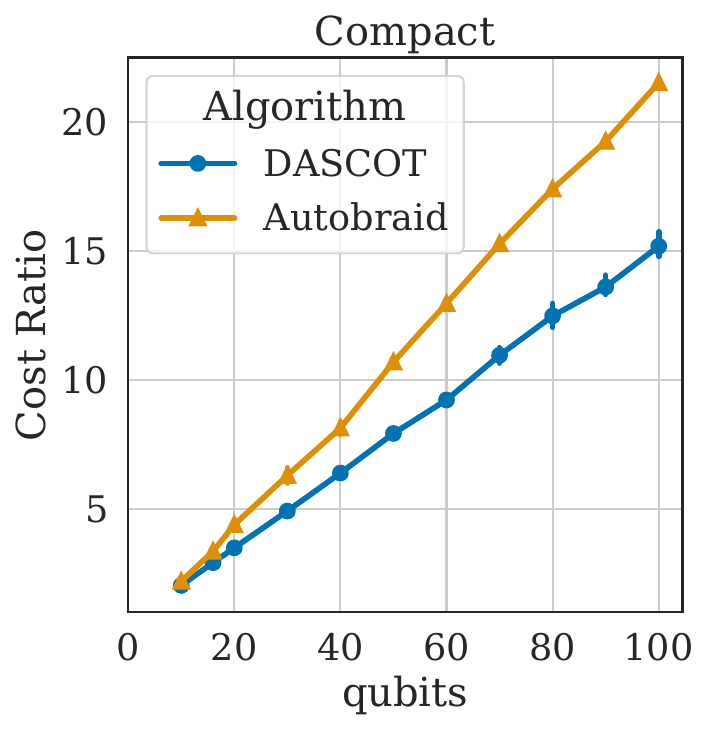}
    \end{subfigure}
  \caption{Cost ratios on QFT Circuits}
  \label{fig:qft}
\end{figure}

\paragraph{Synthetic dense circuits} We also generated a set of synthetic circuits to isolate the effect of density. To construct synthetic
circuits, we iteratively applied a random layer of gates. For the ``high-density'' circuits, a gate is applied
to each qubit in each layer. For the ``low-density'' circuits, a gate is only applied to 1/10 of the qubits
in each layer. We generated circuits with up to 100 qubits
and up to 100 layers. 

\cref{fig:rand_dense} compares the cost ratios obtained by \ours and \autobraid for synthetic 
circuits. Each point represents a circuit. Points above the black line $y=x$ are those where \ours 
produces a better solution than \autobraid. For both algorithms, low-density circuits lead to lower cost ratios than high-density. 
This is expected because there is a high probability of conflict among the many pairs of simultaneous paths in 
the maximally dense layers.
However, \ours is more robust to this complexity, outperforming \autobraid on almost all of the dense
circuits across both architectures. We see a mean reduction in cost ratio of 22\% on the \squaresparse architecture
and 13\% on the \compact architecture.

\begin{figure}
  \begin{subfigure}{0.285\linewidth}
  \includegraphics[width=\linewidth]{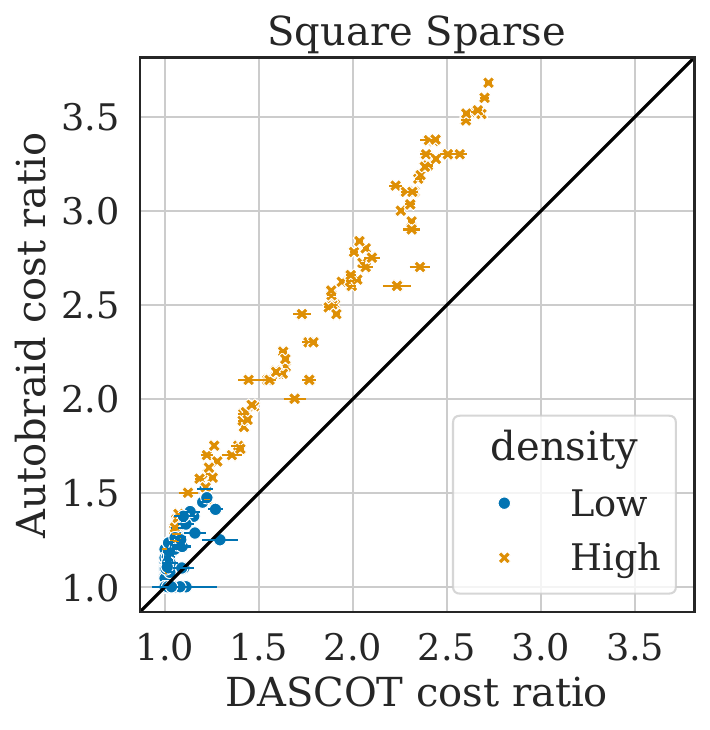}
  \end{subfigure}
  \begin{subfigure}{0.285\linewidth}
    \includegraphics[width=\linewidth]{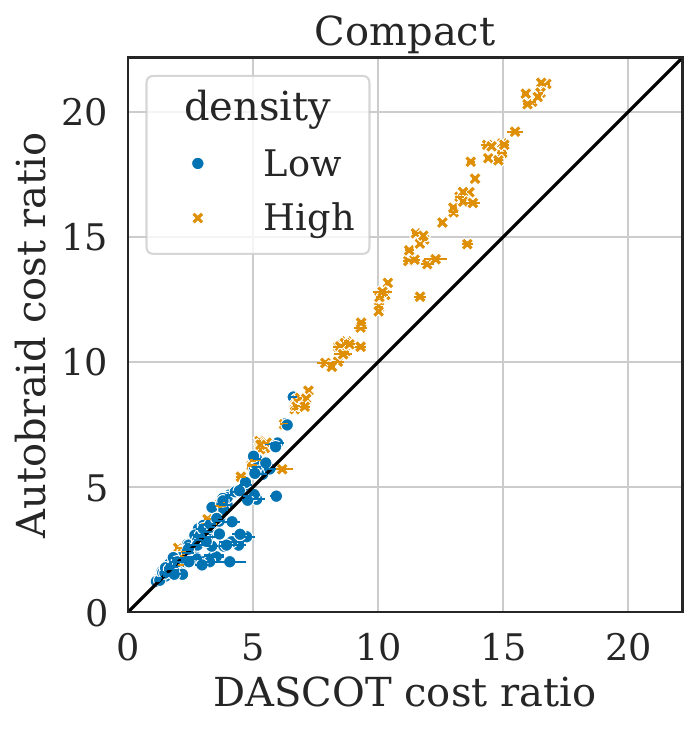}
    \end{subfigure}
  \caption{Cost ratios on random circuits of varying density}
  \label{fig:rand_dense}
\end{figure}
\begin{mybox}
  \textbf{RQ1 Summary:} \ours and \autobraid are both close to optimal on the \squaresparse architecture, while
\ours outperforms the baseline on the \compact architecture, producing a strictly better solution for 59\% of
circuits and matching the Autobraid solution in a further 25\%. 
We see especially significant gains from our
approach on dense circuits, with up to a 31\% improvement on QFT circuits.
\end{mybox}
\subsection{(RQ2) Ablation Study}

Next, we perform an ablation study to isolate the performance contributions of our mapping and routing algorithms. We focus on the \compact architecture and circuits with 8 or more qubits, since these are more difficult instances 
where a naive approach is unlikely to perform well. 

First, to assess our mapping algorithm, we compare against a version of 
our algorithm called \oursrandmap which chooses a qubit map at random, then applies the \ours routing algorithm (\cref{alg:laroute}). 
Choosing a random map is equivalent to setting the initial temperature of the simulated annealing search equal to the final temperature, performing zero iterations of the search,
or not applying a mapping optimization pass at all as in the lattice surgery compiler (LSC) \cite{watkins2023high}.

In \cref{fig:rq2}(left), we compare \ours to \oursrandmap on our application circuits in terms of cost ratio. 
For the majority of circuits (60\%), \ours produces strictly better results. For the remaining circuits, 
the search at the mapping phase does not improve the cost ratio, but there are no examples where 
it leads to a significantly worse cost ratio.

\begin{figure}
  \begin{subfigure}{0.26\linewidth}
    \centering
  \includegraphics[width=\linewidth]{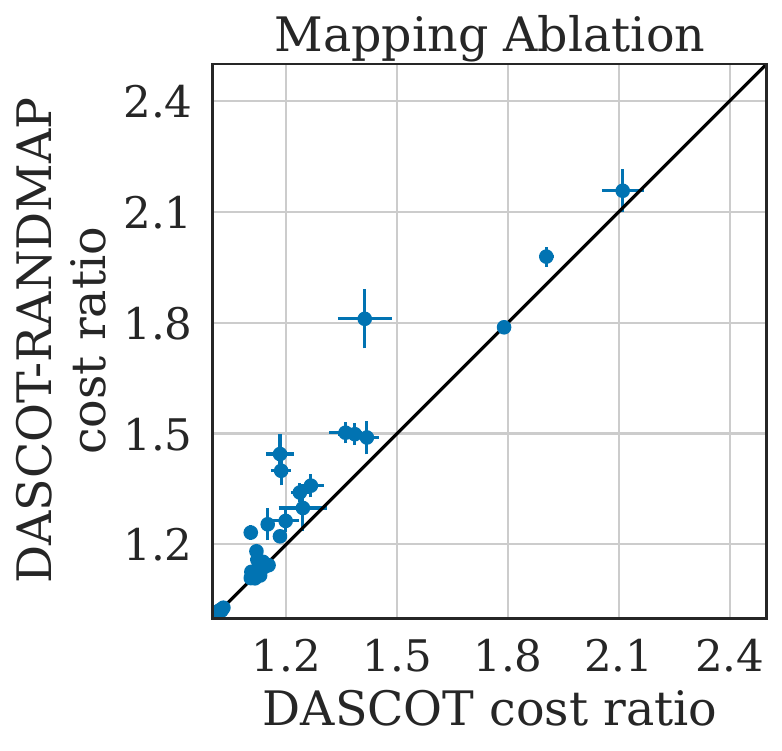}
  \end{subfigure}
  \begin{subfigure}{0.26\linewidth}
    \includegraphics[width=\linewidth]{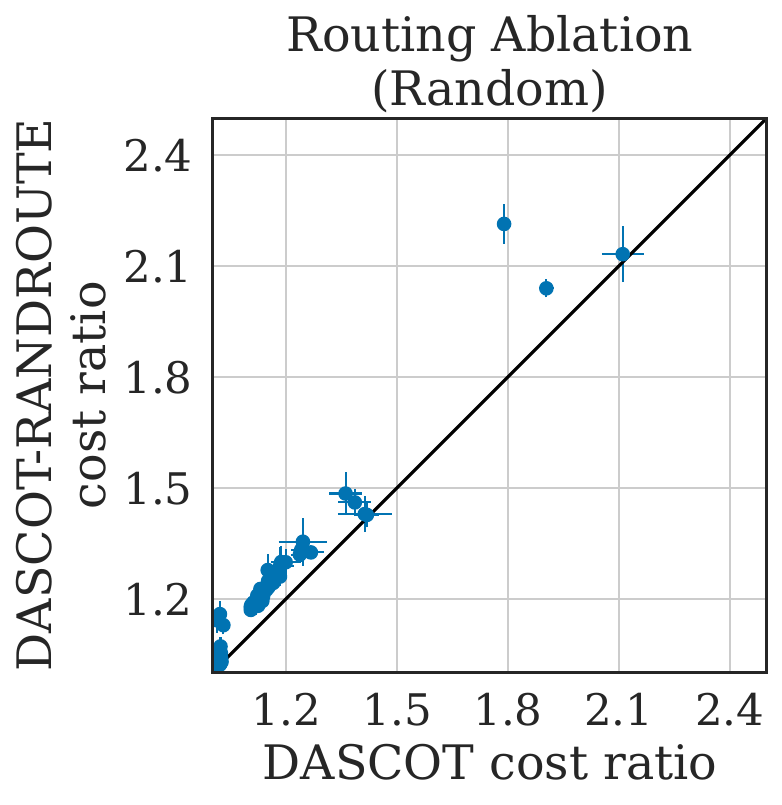}
    \end{subfigure}
    \begin{subfigure}{0.26\linewidth}
      \includegraphics[width=\linewidth]{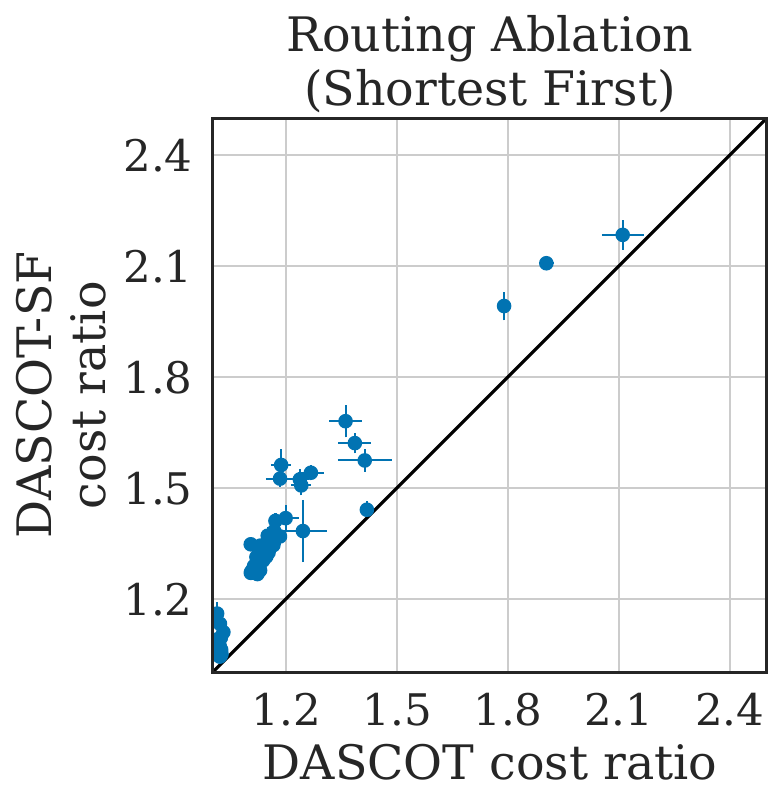}
      \end{subfigure}
  \caption{Effects of mapping and routing }
  \label{fig:rq2}
\end{figure}

We see more uniform cost reduction for the routing phase.
We analogously define the \oursrandroute algorithm which first applies the \ours mapping algorithm (\cref{alg:mapping}). Then, for routing, it replaces
the minimization procedure with a random order selection, once again corresponding to zero iterations of simulated annealing
and the procedure applied by LSC \cite{watkins2023high}. 
In \cref{fig:rq2}(middle), we compare \ours to \oursrandroute. The results indicate
that a dependency-aware analysis of routing order leads to significantly stronger solutions than committing to 
a single random order at each step. We see a strict improvement in mean cost ratio on 90\% of application circuits from applying dependency-aware routing 
for a mean relative improvement of 6\%. Similar findings hold when we compare \ours to choosing a routing 
order by sorting gates according to routing path length and criticality. We show the results for the shortest-first
routing order, which is the approach of EDPC \cite{PRXQuantum.3.020342}, in \cref{fig:rq2}(right). The plot
for the critical-first sorting order closely resembles this one and is included in \cref{sec:plots}.

\begin{mybox}
  \textbf{RQ2 Summary:} Optimization at both stages contributes to the performance of our algorithm.
  However, more consistent gains can be found at
routing phase, where optimizing routing order improves results on 90\% of circuits, and by 6\% on average.
\end{mybox}

\subsection{(RQ3) Comparison to Optimal}
We also assess how well \ours approximates the optimal solution using the baseline described in \cref{sec:opt-solver}. Our implementation
generates the constraints, then queries the \sat solver CaDiCal \cite{BiereFazekasFleuryHeisinger-SAT-Competition-2020-solvers} via the PySAT toolkit \cite{imms-sat18}.  
We focus on the \compact architecture, as \ours generally reaches the depth lower-bound on the \squaresparse architecture, so we know the solutions are 
optimal.

In \cref{fig:rq3}, we compare \ours to the optimal solution in 69 cases where we are able to obtain one within the time and memory bounds.
Circuits are sorted by their \emph{actual cost ratio}, which replaces the theoretical, architecture-independent
lower bound from the cost ratio with the true optimal solution for the \compact architecture. 
For 64 of the 69 circuits,
the mean actual cost ratio is less than 1.25, meaning the mean solution produced by \ours is within 25\% of optimal. A 
common characteristic shared by the five outlier circuits is relatively high density. Such circuits remain 
challenging for \ours despite improvements over the state of the art.
\begin{figure}
  \includegraphics[width=0.257\linewidth]{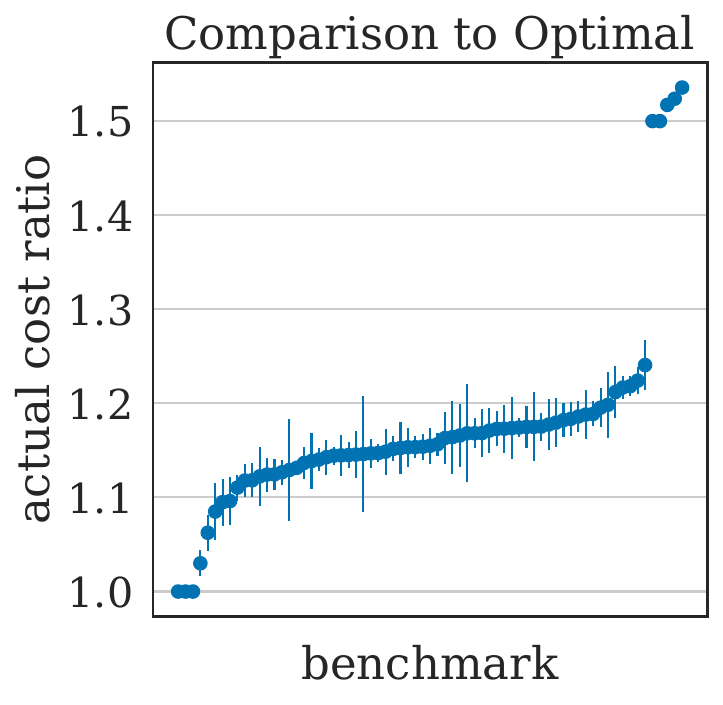}
  \caption{Optimality on the \compact architecture}
  \label{fig:rq3}
\end{figure}
\begin{mybox}
  \textbf{RQ3 Summary}: For small circuits where optimal solving is feasible, we find that \ours finds solutions
that are within 1.25x optimal for all but a few examples.
\end{mybox}

\subsection{(RQ4) Optimality-Scalability Tradeoff}
  We now study how the runtime of our algorithm scales with the complexity of the \scc problem.
  In particular, we are interested in the routing algorithm because mapping can be terminated at any point
  and our ablation study suggests most of the improvements in solution quality are found at the routing stage. 
The complexity of an \scc problem is multidimensional, depending on parameters which include
the number of qubits in the input circuit, the depth of the circuit, and the density of the topological
layers. We highlight two circuit classes with contrasting characteristics:
Bernstein-Vazirani circuits and Quantum Fourier Transform circuits.

Bernstein-Vazirani circuits have low depth and each 
layer is minimal, consisting of a single \cnot gate. Because of this structure, our approach can
route Bernstein-Vazirani circuits with thousands of qubits, as shown in \cref{fig:scaling}(left). 
On the other hand, QFT circuits of the same qubit count have many more layers, with several gates per layer.
Thus, the problem complexity grows much more rapidly with respect to qubit count (see \cref{fig:scaling}(right)), such that the largest 
circuit routed within an hour has 170 qubits on the \compact architecture and 200 qubits on the \squaresparse architecture.

\begin{figure}[h]
  \begin{subfigure}{0.295\linewidth}
  \centering
  \includegraphics[width=\linewidth]{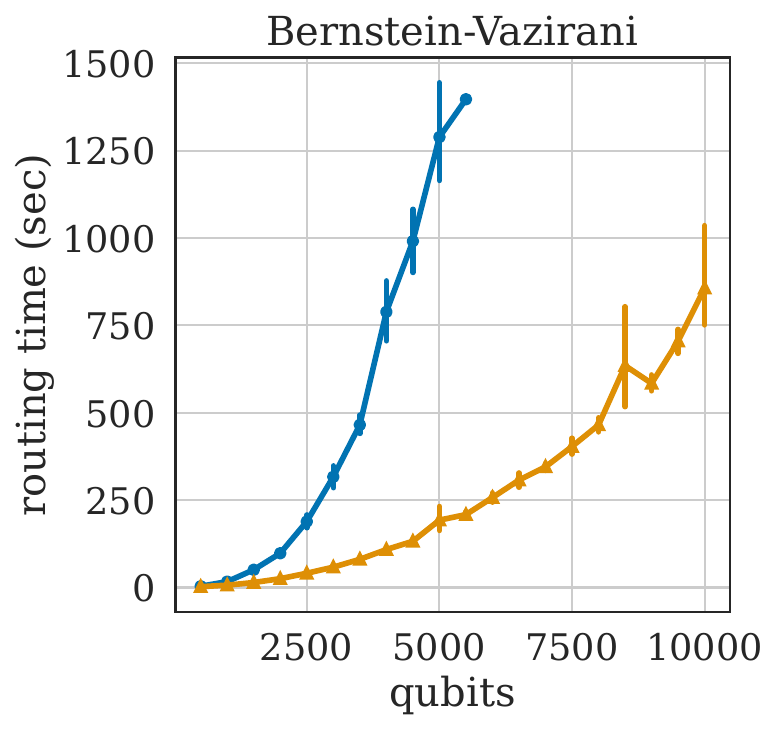}
  % \label{fig:bv-scaling}
  \end{subfigure}
  \begin{subfigure}{0.285\linewidth}
    \centering
    \includegraphics[width=\linewidth]{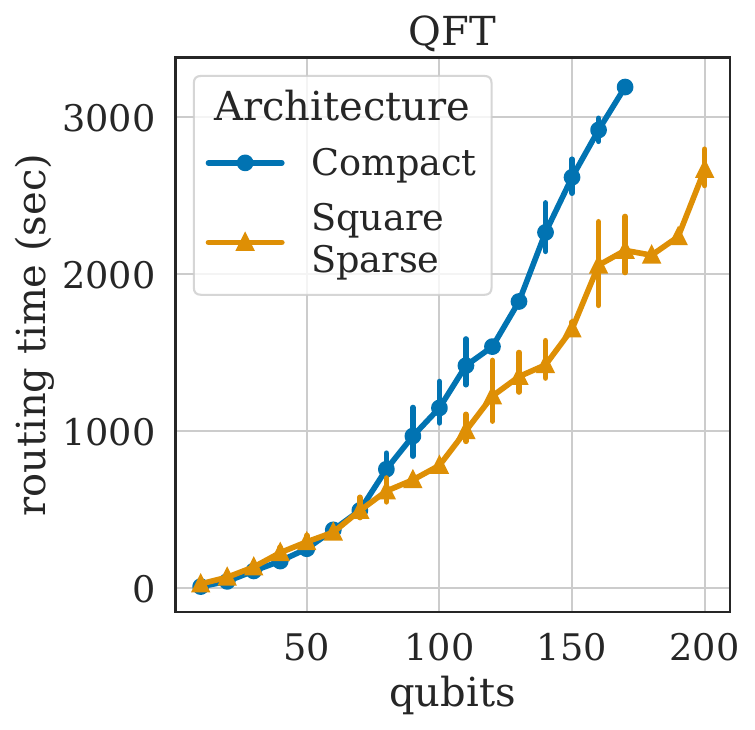}
    % \caption{}
    % \label{qft-scaling}
    \end{subfigure}
  \caption{Runtime scaling of \ours}
  \label{fig:scaling}
\end{figure}

However, our approach is flexible and can terminate more quickly by searching for fewer iterations.
In \cref{fig:tradeoffs}, we compare the solution quality and runtime of three routing algorithms: \ours and \oursrandroute from
our ablation study (the latter is abbreviated to \textsc{rand} in the legend) and an intermediate configuration, \ourslimited, which performs 1/2 as many search iterations as \ours.
Benchmarks are sorted along the $x$-axis in ascending order with respect to \ourslimited.
We see that \ourslimited finds solutions with cost ratios near the midpoint between \oursrandroute (2\% better than
\oursrandroute and 3\% worse that \ours on average) while terminating 41\% faster than \ours on average.
\begin{figure}[h]
  \begin{subfigure}{0.285\linewidth}
  \centering
  \includegraphics[width=\linewidth]{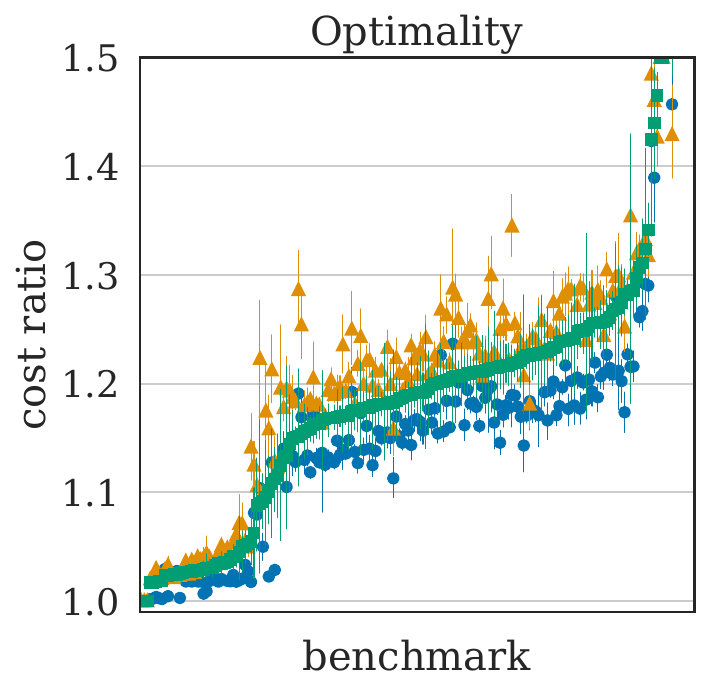}
  \end{subfigure}
  \begin{subfigure}{0.3\linewidth}
    \centering
    \includegraphics[width=\linewidth]{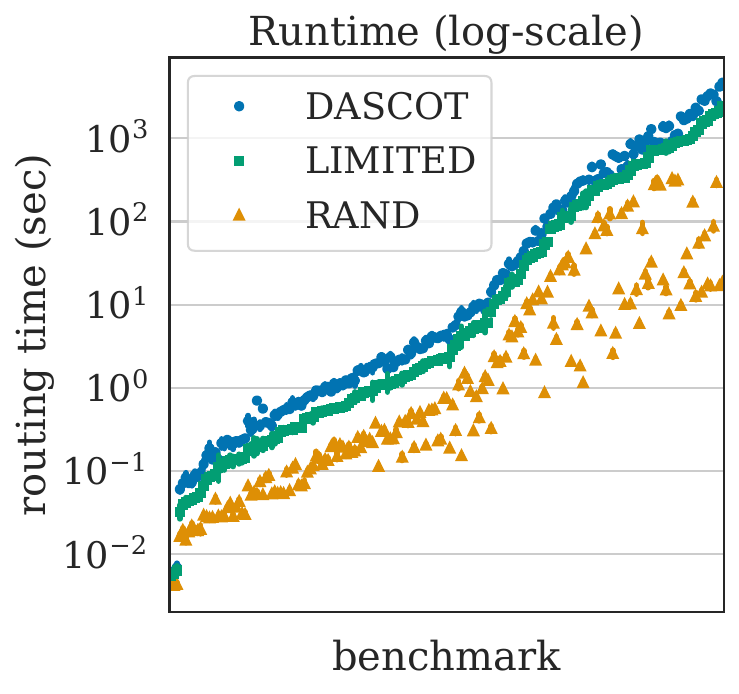}
    \end{subfigure}
    \caption{Trading compute for optimality}
  \label{fig:tradeoffs}
\end{figure}

\begin{mybox}
  \textbf{RQ4 Summary:} \ours can be applied to Bernstein-Vazirani circuits with up to 10,000 qubits and 
  QFT circuits up to 200. Additionally, applying a limited version of \ours decreases run time by 41\%, incurring
  a only a 3\% cost in terms of solution quality.
\end{mybox}

\section{Discussion and Related Work}

\paragraph{Other allocation and disjoint paths problems}
Surface code mapping and routing is an assignment of program variables to limited physical resources such that the runtime efficiency is maximized, much like the classical problem of register allocation \cite{COOPER2023663,chaitinregalloc,polettoregalloc}. Disjoint paths problems like the constraints on a valid gate route have also been studied in other contexts
including VLSI routing \cite{wireRouting, ndpvlsiaggarwal} and all-optical networks \cite{Yonatan-yuval-Optical-routing,effallopt,effalloptagg}.
Inspired by these applications, the study of approximation algorithms for disjoint path problems is an active area with both hardness results \cite{1530717,DBLP:series/lncs/Erlebach06} and approximation algorithms on planar and grid graphs \cite{10.1145/2897518.2897538, chuzhoy_et_al:LIPIcs:2015:5303}. However, a distinguishing feature of this particular  disjoint path problem
is the dependency between paths,  forcing the routing of certain pairs in a particular order.

\paragraph{NISQ qubit mapping and routing}
    A wide array of  work has targeted the qubit mapping and \cnot routing problem for noisy intermediate-scale quantum (NISQ) computers
 without error-correction. In the NISQ setting, \cnot gates are only executable between qubits which are mapped to adjacent vertices on the architecture graph.
 NISQ mapping and routing tools use \swap gates to update the qubit map, moving qubits to adjacent locations.   
The constraints of \scc are fundamentally different: we can perform two qubit gates between any pair
of qubits in constant time as long as there is an appropriate routing path available. Nevertheless, we 
can leverage certain insights developed in the NISQ setting. For example, our layered interaction graph approach for mapping is an extension of NISQ algorithms those that construct qubit maps based on interaction graphs \cite{siraichi2019qubit, cowtan2019qubit}. 
Also, our \sat encoding for the optimal baseline shares structure with encodings of the NISQ problem in \sat, \textsc{smt}, or \textsc{maxsat} \cite{Wille2019MappingQC,olsq,Lin2023ScalableOL,Molavi2022QubitMA,satmapper,shaik_et_al:LIPIcs.SAT.2024.26}, though the constraints describing valid
paths are unique to \scc and thus to our encoding.

\paragraph{Defect-braid routing} Some recent work in compilation for surface code quantum
devices has addressed similar problems to the surface code mapping and routing problem. \citet{Javadi-abhari-17}
identifies \cnot contention as an important architectural design factor and develops
 heuristics for mapping qubits and scheduling \cnot gates implemented as defect braids.
 \autobraid \cite{autobraid} builds on this work with a new algorithm for mapping that 
 minimizes the number of large groups of \cnot gates with overlapping bounding boxes and a stack-based routing algorithm. However, neither considers \T gates,
 multiple architectures, or the constraints of lattice surgery \cnot gates.
 
Autobraid includes a procedure for inserting \textsc{swap} operations
 to shuffle the qubit map when a small proportion of front layer of gates can be routed in the same time step.
 In principle, such shuffling can be incorporated into \ours. However, \swap gates are expensive,
 decomposing into three \cnot gates of alternating orientation (the control of one \cnot is the target of the next). Because \cnot gates of opposite orientation must be routed along different paths, \swap gates occupy many vertices and are
 in parallel with other gates. Thus, transitioning to a new mapping often has a net negative effect as the added steps for \swap gates outweighs the benefit of the new mapping.
 
 \paragraph{Lattice surgery compilation} For the lattice surgery setting, LSC \cite{watkins2023high} provides a scalable framework for translation of quantum 
gates to lattice surgery operations, but does not optimize the mapping or routing solution.
EDPC \cite{PRXQuantum.3.020342} considers a shortest-first routing algorithm, but assumes a given fixed map 
 that places qubits in a ``sparse'' structure.  None of the above work provides an optimal solver for 
 the compilation problem they address. \citet{Lao_2018} does solve a mapping and routing problem optimally with integer linear programming, but this formulation does not permit long-range \cnot gates. 
 LaSsynth \cite{Tan:2024cfn} uses a \sat solver to optimize lattice surgery representations of small subroutines (5-20 qubits and 10-100 operations) 
at a lower level of abstraction.
 
 An alternative proposed compilation pipeline serializes a circuit to a sequence of Pauli product rotations \cite{Litinski2018AGO,Beverland:2022rpv}.
 In this setting, TopQAD \cite{mqlss} applies an algorithm similar to LSC: performing no mapping optimization and routing gates in a random order.
 The only difference between this greedy routing algorithm and \oursrandroute is in the specifics of the routed operations. Pauli product rotations are implemented via routing trees with an interacting qubit at each leaf, rather than routing paths with an interacting qubit at each endpoint.
 However, sequential Pauli compilation can lead to prohibitively high run times, converting a circuit with $k$ parallel operations to a compiled  
 circuit of depth $k$ \cite{PRXQuantum.3.020342}.

\section{Conclusions}
In this paper, we have tackled the surface code mapping and routing problem, a critical
problem in compilation for emerging practical quantum computers. We developed an algorithm to deliver high-quality solutions in reasonable
time via simulated annealing and dependency-awareness. In future work, we plan to extend our model and algorithms to capture a broader class of quantum architectures.
For example, we could support recent designs for heterogeneous multi-chip fault-tolerant quantum architectures \cite{chiplet,hetarch}. 
The mapping and routing problem in this case is nuanced because not all \cnot gates are equivalent. A cross-chip \cnot gate is possible, but more costly.
\section*{Data-Availability Statement}
The software and benchmark suite that supports our evaluation in \cref{sec:eval} is available in a public archive \cite{artifact}. This software artifact includes a Docker image packaging the source code for \ours and the baseline algorithms,
along with scripts for reproducing our empirical results and generating plots matching those included in the paper.

\section*{Acknowledgments}
We thank the anonymous reviewers for their insightful feedback and suggestions. We are grateful to the CHTC team for 
maintaining the computing resources which enabled our empirical evaluation \cite{chtc}.
This work is supported by \textsc{nsf} grants \#1652140 and \#2212232 and awards from Meta and Amazon. This research is also partially supported by the OVCRGE at the University of Wisconsin-Madison with funding from the Wisconsin Alumni Research Foundation.

\bibliographystyle{ACM-Reference-Format}
\bibliography{references}
\newpage
\appendix
\section{Hardness of Surface Code Mapping and Routing}
\label{sec:hardness}
In this section, we prove that surface code mapping and routing is $\np$-complete via a reduction from
a multiprocessor scheduling problem.
To that end, we state the decision version of the \scc problem below.
\begin{mybox}
\begin{definition}[\scc Decision Problem]\ \\
  \indent Given: An architecture $\arch$ and circuit $\circuit$ and time limit $t_s$  \\ 
  \indent Question: Is there a qubit map $\map$ and valid gate route ($t$, $\rspace$, $\rtime$) for $\arch$, $\circuit$, and $\map$ with $t \leq t_s$?
\end{definition}
\end{mybox}

\begin{theorem}\label{thm:scmr_np_complete}
  The \scc Decision Problem is \np-complete.
\end{theorem}

\subsection{Overview of the Proof of \cref{thm:scmr_np_complete}}
Clearly, \scc is in $\np$, with a polynomial length certificate consisting of a description of the qubit map and a sequence of time steps representing a valid gate route (alternatively, our poly-sized \sat encoding in \cref{sec:opt-solver} is a constructive proof). The hardness proof is a reduction from the $k$-processor scheduling problem, shown to be \np-complete by  \citet{ULLMAN1975384} and stated below.

\begin{mybox}
\begin{definition}[Processor Scheduling Problem]\ \\
  \indent Given: a finite partially ordered set of jobs $(\jobset, \prec)$, a number of processors $k$, and a time limit $t_p$  \\
  \indent Question: Does there exist a schedule $\sched$ from $\jobset$ to the set $\{1, \ldots, t_p\}$ such that
  \begin{itemize}
    \item if $\job \prec \job'$  then $\sched(j) < \sched(j')$;
    \item for each $1\leq i \leq t_p$, the subset of $\jobset$ mapped to $i$ by $\sched$ has cardinality at most $k$?
  \end{itemize}
\end{definition}
\end{mybox}

We will abbreviate ``Processor Scheduling Problem'' as \psp and
denote an instance with a job set $(\jobset, \prec)$, $k$ processors, and time limit $t_p$ by
$\psp(\jobset, k, t_p)$ (leaving the order implied). 
\psp is a natural choice for this reduction since it also involves
executing tasks with dependencies as a sequence of time steps, but without the space constraints of \scc. 
Just as a job occupies a processor for a time step, 
a \T gate occupies a magic state qubit for a time step. Therefore the overall strategy of the proof is to encode the dependency structure of a \psp instance into the \T gates of a circuit, and
include the same number of magic state qubits as processors to match the simultaneous capacity. Before providing the general proof, we will walk through this process with the example job set shown in \cref{fig:hard_psp}.
This instance consists of four jobs: one has no
prerequisites and the other three depend solely on it. 
Our goal is the following:

\begin{center}
  \emph{Associate each job in a \psp instance with a \T gate, such that the dependency structure of the resulting circuit 
is exactly that of the \psp instance.} 
\end{center}

There are three roadblocks along this path which motivate the three major components of the reduction.

\subsubsection*{Challenge 1 (High-degree jobs):} Since quantum gates act on at most two qubits, we cannot directly encode jobs 
         with more than two dependencies, such as Job A in our example set.

 \emph{Solution}:  Represent each job with a subcircuit called a job gadget containing exactly one \T gate 
 and several \cnot gates. The job gadget representing the jobs in our example is shown in \cref{fig:job_gadget_circ}.
The \T gate of the job gadget is applied to the first qubit. The other qubits are used to induce dependencies
between \T gates. In our example, the maximum number of incoming or outgoing dependencies for a job is 
three, so we have three of these additional ``input/output'' (I/O) qubits. Notice that if we construct a circuit where a job gadget is preceded
by a gate $g$ on one of its I/O qubits, then the \T gate depends on $g$ transitively via one of the first three \cnot gates in the job gadget. Likewise,
in a circuit where a job gadget is followed by a gate $g$ on one of its I/O qubits, $g$ depends on the \T gate through one of the last three \cnot gates in the job gadget. 
We use this property of the job gadget to induce the appropriate dependencies between job gadgets in a \emph{dependency circuit}. In the dependency circuit,
we include a \cnot between I/O qubits of job gadgets for each direct dependency of jobs. Parallel dependencies (like A to B and A to C in our example) use different I/O qubits.
The full dependency circuit corresponding to this example is in \cref{fig:dep_circ_itself}.

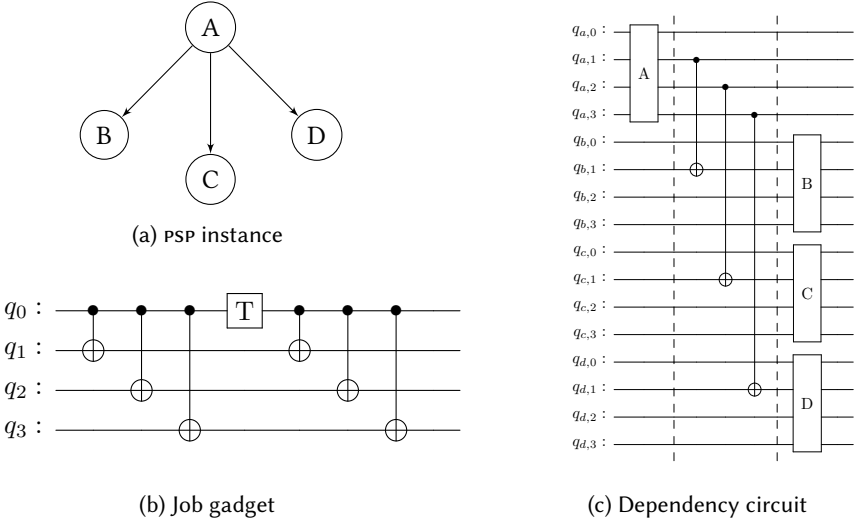
\begin{figure}
  \begin{tcbitemize}[raster equal height=rows,
  raster every box/.style=blankest]
  \hspace{1em}
  \tcbitem
      \begin{tcbitemize}[raster columns=1]
          \tcbitem \centering \subcaptionbox{\psp instance\label{fig:hard_psp}}{ \begin{tikzpicture}[main/.style = {draw, circle}, edge/.style={->,> = latex'}, node distance=2cm]
            \node[main] (1) {A};
            \node[main] (2) [below left of=1] {B};
            \node[main] (3) [below of=1] {C};
            \node[main] (4) [below right of=1] {D};
            \draw[edge] (1) -> (2);
            \draw[edge] (1) -- (3);
            \draw[edge] (1) -- (4);
          \end{tikzpicture}}
          \vfill
          \tcbitem \subcaptionbox{Job gadget\label{fig:job_gadget_circ}}{\includegraphics[width=\linewidth]{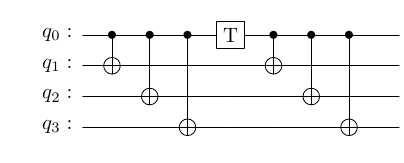}}
        \end{tcbitemize}
      \tcbitem
      \hspace{2em}
      \subcaptionbox{Dependency circuit\label{fig:dep_circ_itself}}{\includegraphics[width=0.615\linewidth]{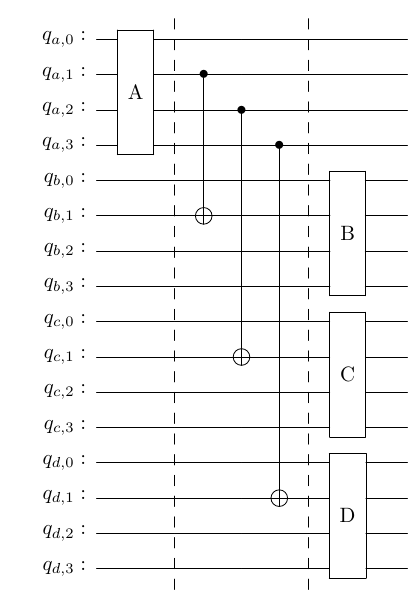}}
  \end{tcbitemize}
  \caption{Constructing the dependency circuit for our running example. Each rectangle in the dependency
  circuit shown in (c) is one copy of the circuit in (b)}
  \label{fig:dep_circ_full}
\end{figure}
\subsubsection*{Challenge 2 (Time Dilation):} We need a longer time limit for our \scc instance than the original \psp instance. 
Consider the instance consisting of our running job set, $k=3$ processors, and the time limit $t_p=2$. Clearly, the answer is ``yes'': execute
job A in step 1 and the other three jobs in step 2. Two steps is not sufficient for executing the dependency circuit because we have ``stretched''
each job into several gates and added ``transition'' \cnot gates between job gadgets.

\emph{Solution:}
Choose an extended time limit for the \scc instance with both a stretching and transition term. In our case, each job gadget includes
a chain of 7 dependent gates, so the stretching term is $7t_p$. Since the maximum number of
dependencies of a job is 3 and we have $k$ independent processors, we need to be able to execute up to $3k$ transition \cnot gates between jobs. 
This yields a transition term of $3k(t_p-1)$, for a total time limit of $t_s = 7t_p + 3k(t_p-1)$.

\subsubsection*{Challenge 3 (Staggered Execution):} When we choose a time limit that accounts for the stretching and transitions, we over-correct and introduce spurious
solutions. Using the same set of jobs and time limit of $t_p=2$, the answer is ``no'' with only $k=2$
processors. The formula in the previous paragraph yields $t_s = 20$ time steps to execute the dependency circuit. The dependency circuit \emph{can} be executed within 20 steps, even with a magic state capacity limit of 2.
This can be done by delaying the execution of job gadget D by one time step. Then, only
two \T gates are executed simultaneously (for B and C), and the job set is executed in 18 total steps.
Such a solution breaks the correspondence between \psp and \scc.

\emph{Solution}: 
Eliminate the unwanted solutions by adding an independent circuit called the \emph{cycle circuit}. 
The cycle circuit applies $t_s$ gates to each qubit, so it can be executed in $t_s$ steps only when each gate is executed as soon as theoretically possible,
at the time step equal to its depth in the circuit. This is useful because it allows us to restrict the execution of
\T gates in the dependency circuit to specific time steps. The cycle circuit  
occupies and releases the magic state qubits in a repeating cycle. The number of cycles is equal to the \psp time limit $t_p$.
The \T gates from the dependency circuit are executable only in the steps that lie at the center of a cycle (where the cycle circuit executes a \cnot), and there are $t_p$ such steps.
Now, it is no longer possible to delay the execution of job gadget $D$ by one time step because no magic state qubits are available.

\cref{fig:with_cycle_circuit} shows the dependency circuit extended with the cycle circuit. 
One pair of qubits is shown, corresponding to a \psp instance with a single processor. To encode $k$ processors, 
we include $k$ copies of the two-qubit circuit.

\begin{figure}[h!]
  \includegraphics[width=0.6\linewidth]{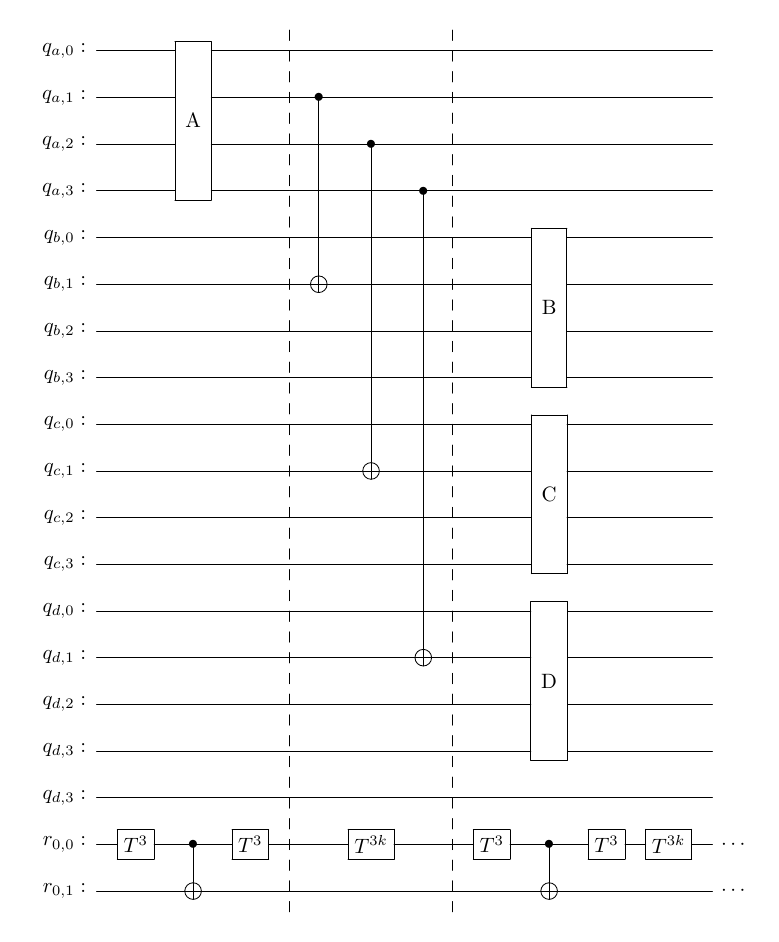}
  \caption{Full circuit in our reduction. The notation $T^x$ means $x$ applications of the \T gate}
  \label{fig:with_cycle_circuit}
\end{figure}

\subsubsection*{The architecture}
Since the $\T$ gates in the dependency circuit correspond to jobs, we design our target architecture to match the capacity
of a \psp instance with $k$ processors by including $k$ magic state qubits, so that $k$ $\T$ gates can
be executed simultaneously. Other than this capacity limit, we do not want to introduce any additional
limitations. If two jobs can be assigned to the different processors at the same time step,
their job gadgets should be simultaneously executable. 
We ensure this property by organizing our architecture into $k$ ``processor units.'' Each processor unit contains one magic state qubit.
 We can connect processor units into a connected grid arbitrarily. For simplicity, we choose to chain them 
linearly, forming a ``long'' architecture with 4 rows and many columns. 

\subsubsection*{Summary} Given a \psp instance $(\jobset, k, t_p)$, we can construct an equivalent \scc instance $(\arch, \circuit, t_s)$
where $\arch$ consists of $k$ processor units, $\circuit$ is the union of the dependency circuit and the cycle circuit, and $t_s$ is 
computed as described above from the number of processors and the maximum number of dependencies of a job. We can translate solutions from one instance to another by executing a job in step $i$ on processor $j$ if and only if the
corresponding job gadget is executed during cycle $i$ of the cycle circuit and the job circuit qubits are mapped to processor unit $j$. 
The rest of this section elaborates on each of these proof steps.

\subsection{The Dependency Circuit}
We begin with the \emph{dependency circuit}: a structure that encodes the partial order on jobs as a circuit.
This is nontrivial because the dependency structure allowed in a \psp instance is more general that
that of quantum circuits. 

To simulate an arbitrary dependency structure under this limitation,
we use a construction called a \emph{job gadget}. 
The job gadget is designed such that the \T gate qubit inherits the incoming and outgoing dependencies
of all of the other qubits. In general, the job gadget for a \psp instance $\jobset$ includes $d+1$
qubits, where $d$ is the maximum degree of any job in the Hasse diagram of $\jobset$.
One of these qubits is used to execute a $\T$ gate that represents the job itself, and the rest
are used to introduce dependencies with other job gadgets.
Let $q_{\job, {i}}$ denote the $i$th qubit in the job gadget for $\job$. The job gadget consists of the gates
\cnot $q_{\job, 0}$ $q_{\job, i}$ for each $1 \leq i \leq d$, then the gate \T $q_{\job, 0}$, and finally a repetition of the gates \cnot $q_{\job, 0}$ $q_{\job,i}$
for each $1 \leq i \leq d$.

For each job $\job$ we assign an indexing to the incoming and outgoing edges. Then,
for each edge $e = (\job, \job')$
we include the gate \cnot $q_{\job, i}$ $q_{\job', i'}$ between the job gadgets for $\job$ and $\job'$, where $i$ is the index of $e$
as an outgoing edge of $\job$ and $i'$ is its index as an incoming edge of $\job'$. Notice that
this ensures the gate \T $q_{j,0}$ must be executed before the gate \T $q_{\job',0}$ as desired. 
We let
$\depcirc(\jobset)$ denote the dependency circuit for a job set $\jobset$ and state this formally as 
\cref{thm:dependency_circuit}, where  $\prec^*$ is the transitive closure of the relation $g \prec_C g'$ on gates in a circuit which says $g'$ depends on $g$.

\begin{restatable}{lemma}{depcirclemma}
  \label{thm:dependency_circuit}
  Let $(\jobset, \prec_J)$ be a partially ordered set. For each pair of jobs $\job, \job' \in \jobset$, the $\T$ gate \T $q_{\job, 0}$ in $\depcirc(\jobset)$ satisfies the property that \T $q_{\job,0} \prec^* \T ~ q_{\job',0}$ if and only $\job \prec_J \job'$.
\end{restatable}
\begin{proof}
First suppose $\job \prec_J \job'$. We can induct on the length of the path from $\job$ to $\job'$ in the Hasse diagram for $\jobset$.
If the length of the path is 1, there are \cnot gates $g_1 = \cnot ~ q_{\job, 0} ~  q_{\job, i}$, $g_2 = \cnot ~ q_{\job, i} ~ q_{\job', i'}$, and $g_3 = \cnot ~ q_{\job', 0} ~ q_{\job', i'}$,
such that $\T ~ q_{\job,0} \prec_C g_1, \prec_C g_2 \prec_C g_3 \prec_C \T ~ q_{\job',0}$. Thus, \T $q_{\job,0} \prec^* \T ~ q_{\job',0}$.
If the length of the path is greater than 1, with final edge $(\job_k, \job')$, then by the above argument,
\T $q_{\job_k,0} \prec^* \T ~ q_{\job',0}$, and by induction, \T $q_{\job,0} \prec^* \T ~ q_{\job_k,0}$.
By transitivity, \T $q_{\job,0} \prec^* \T ~ q_{\job',0}$.

Now assume \T $q_{\job,0} \prec^* \T ~ q_{\job',0}$. Then there is a chain $\T ~ q_{\job,0} \prec_C g_1 \prec_C \ldots \prec_C g_k \prec_C \T ~ q_{\job',0}$.
We can similarly proceed by induction on the number of $\T$ gates among the gates $g_1 \ldots g_k$. If there are none, then $(j, j')$ is an edge in the Hasse diagram for $\jobset$. 
Now suppose there are $k$ such gates, with the last one between
a qubit in the job gadget for job $j_k$ and the job gadget for $j'$. Then, by the base case,
there is an edge $(j_k, j)$ in the Hasse diagram, by induction, $j \prec_J j_k$ and by transitivity $j \prec_J j'$.
\end{proof}

\subsection{The Cycle Circuit}
\paragraph{Choosing the time limit}
If $d$ is the maximum in-degree or out-degree of any job, we need $2d+1$ time steps to execute a job
gadget. First, we execute $d$ \cnot gates, then the \T gate, then $d$ more \cnot gates. Therefore, each
time step in a \psp instance will need to be ``stretched'' into $2d+1$ time steps in a corresponding \scc instance. 

Additionally, between time steps $t$ and $t+1$, we must perform up to $dk$ \cnot gates connecting the 
job gadgets executed in step $t$ to those that depend upon them. These cross the architecture, so
in the worst case we need to execute them sequentially, for $dk$ time steps per transition. There are $t_p-1$ transitions between $t_p$ time steps. 
Thus, we set the time limit $t_s = (2d+1)t_p + (dk)(t_p-1)$.

\paragraph{The cycle circuit}
Extending the time limit as described above preserves \psp solutions, but it also could introduce spurious
solutions to the \scc instance that involve executing \T gates ``between'' time steps and have no corresponding \psp schedule. The problem is that
stretching the time steps has allowed squeezing multiple \T gates into different steps in a range
intended to represent a single \psp step.

The next key component of the proof is another circuit \textbf{on a distinct register of qubits} designed to
disallow such solutions. We call this the \emph{cycle circuit} because it consists of a repeating pattern of $\T$ gates to occupy the magic state qubits, so that there are only $t$ steps in which $\T$
gates from the dependency circuit can be performed. For each of the $k$ processors, we include $t$ cycles. 
Each cycle consists of the following pattern:
\begin{itemize}
  \item $d$ \T gates 
  \item A \cnot gate 
  \item $d$ more \T gates.
\end{itemize}

Between cycles we include $dk$ \T gates as well. We denote the cycle circuit with parameters
$d$, $k$, and $t_p$ as $\cyccirc(d,k,t_p).$ \cref{thm:cyclecirc} says that we can use the cycle circuit to restrict the dependency circuit to executing \T gates during the $t$
time steps where the cycle circuit executes \cnot gates.

\begin{restatable}{lemma}{cyclecirc}
  \label{thm:cyclecirc}
  Any scheduling for $\cyccirc(J,k,t_p)$ in $t_s = (2d+1)t_p + (dk)(t_p-1)$ time steps
  includes $t$ time steps where no \T gates are executed with $dk$ time steps where $k$ \T gates are executed interspersed between each.
\end{restatable}

\begin{proof}
  Since each gate is part of a chain of $t_s$ dependent operations, it must be executed at
  the time step corresponding to its position $i$ in the chain. It cannot be executed in an earlier time step
  because either $i=1$ or there are $i-1$ steps necessary to execute the gates which precede it. It cannot be executed in a later time step because then
  the last gate in its chain will not be executed within the $t_s$ step limit. There are $t_s$ positions
  where all $k$ chains of dependent operations in $\cyccirc(J,k,t_s)$ include a \cnot gate, and thus $t$ time steps where
  $k$ \cnot gates are executed (and no \T gates). Moreover, these positions are separated by $dk$ positions where each chain includes a \T gate, for $k$ total simultaneous \T gates.
\end{proof}\

\begin{figure}
  \centering
  \leavevmode

    \Qcircuit @C=1em @R=.7em {
  & \gate{T^{d}} & \ctrl{1} & \gate{T^{d}} & \gate{T^{dk}} & \gate{T^{d}} &\ctrl{1} &  \gate{T^{d}} & \gate{T^{dk}} & \qw  &  \cdots && \gate{T^{d}}& \ctrl{1} & \gate{T^{d}} & \qw \\
  & \qw          &  \targ   & \qw          &\qw            & \qw          &\targ    & \qw           & \qw    & \qw & \cdots   && \qw& \targ    & \qw & \qw \\
  }
  \caption{Fragment of the cycle circuit. There is one layer of \cnot gates for each time step in the \psp instance and one copy of this two-qubit circuit per processor. The notation $\T^{x}$ means $x$ copies of the \T gate.}
  \label{fig:cyclic_circ}

\end{figure}
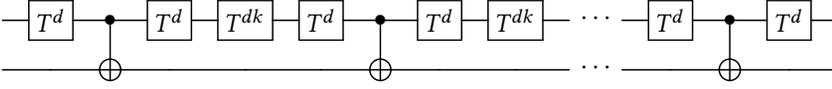
\subsection{Architecture}
\label{sec:arch_details}
The processor unit for this example is shown in \cref{fig:proc-unit}.
The yellow vertex (third row, second column from the right) is the magic state qubit. The other colored vertices represent the mapping that will be
used in one direction of the proof of the correctness of the reduction.
Along the top row we map the 4 qubits of each job gadget together. There are as many such slots as jobs,
to allow as many jobs to be mapped to this processor as desired.
The two pink qubits (near the magic state qubit) are a pair from the
cycle circuit.  In total, a processor unit is a $4 \times 6|J|+1$ grid (the additional column is for connecting processor units) with one magic state qubit in the second row from the bottom and second column from the right.

\begin{figure}
  \includegraphics[width=\linewidth]{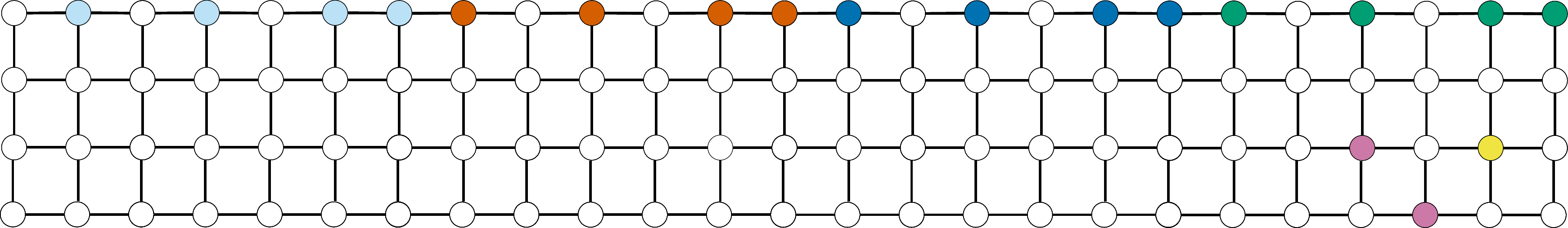}
  \caption{Processor unit for our running example}
  \label{fig:proc-unit}
\end{figure}

\subsection{Putting it all together}
We are now ready to describe the full reduction from \psp to \scc. Let
$\psp(\jobset, k, t_p)$ denote a \psp instance and let $d$ be the maximum in-degree or out-degree of any
job in $\jobset$. We construct the corresponding \scc instance, $\scc(\arch,\circuit, t_s)$
as follows. Set $\arch$ to consist of $k$ processor units. Set $C$ to be composed of two parallel, independent subcircuits: $\depcirc(J)$ acting on a qubit set $Q$
and $\cyccirc(\jobset,k,t)$ acting on a distinct set $Q'$. Set $t_s = (2d+1)t_p + (t_p-1)(dk)$.

Suppose the answer to $\psp(\jobset, k, t_p)$ is ``yes.'' Let $s$ be a corresponding schedule. We can
construct a solution $s'$ for $\scc(\arch, \circuit, t_s)$ from $s$. For the mapping, we map the qubits from a job
gadget to  the processor unit corresponding to the processor on which the job is
executed in $s$. One pair of qubits from the
cycle circuit is mapped to each processor unit as well.  A valid gate route
for the cycle circuit under this mapping can be constructed by executing each gate at the time step equal to its depth. By \cref{thm:cyclecirc}, this
leaves $t_p$ time steps to execute the $\T$ gates of the dependency circuit.
A valid gate route for the dependency circuit can be derived by executing the \T gates in the order described by $s$. That is,
if a job is executed in step $\sigma$ in $s$, then the corresponding \T gate is executed in step $\sigma$
among the $t$ available (the \cnot gates can be executed in the other time steps as described in the time
limit discussion). Since $s$ is a valid $\psp$ solution and the dependencies among these gates
match the $\psp$ instance by \cref{thm:dependency_circuit}, this describes a valid gate route.

If the answer to the $\scc(\arch,\circuit, t_s)$ is ``yes'', we can invert this process and construct a schedule $s$ for
$\psp(\jobset, k, t_p)$. We execute jobs in the order that their \T gates are executed. To choose the processors
for jobs, we fix a bijection between magic state vertices and processors, and execute the job on the processor
corresponding to the magic state vertex used to execute the \T gate. By \cref{thm:cyclecirc}, there are $t_p$ time steps
in such a schedule. By \cref{thm:dependency_circuit}, it will satisfy the dependencies in $\jobset$.
Finally, since there are only $k$ magic state qubits on $\arch$, it must satisfy the requirement that only
$k$ jobs are executed simultaneously.
\section{The Routing Problem is \np-Complete}
\label{sec:routing-np-complete}
In fact, the routing problem is \textsc{np}-complete even for
a given fixed map.
\begin{mybox}
\begin{definition}[Surface Code Routing Decision Problem]\  \\
  \indent Given: An architecture $\arch$ and circuit $\circuit$, a qubit map $\map$, and time limit $t_s$ \\
  \indent Question: Is there a valid gate route ($t$, $\rspace$, $\rtime$) for $\arch$, $\circuit$, $\map$ with $t \leq t_s$?
\end{definition}
\end{mybox}
\begin{restatable}{theorem}{routing}
  The Surface Code Routing Decision Problem is \textsc{np}-complete.
  \label{thm:routing-np-complete}
\end{restatable}

Consistent with notation from \cref{sec:prob_def}, we denote the surface code routing instance with architecture $\arch$, circuit $\circuit$,
and qubit map $\map$
by $\scr(\arch, \circuit, \map)$. Once again, \textsc{scr} is in $\np$, with a sequence of time steps representing a valid gate route serving
as a polynomial length certificate. We prove the hardness part of \cref{thm:routing-np-complete} by noting \scr is closely related to an \np-complete problem
originally considered in the context of VLSI layout \cite{wireRouting}. 

\begin{mybox}
\begin{definition}[Node-Disjoint Paths on a Grid Problem]\ \\
  \label{def:ndp}
  \indent Given: A grid graph $G = (V, E)$ and a set of vertex pairs $\mathcal{R} = \{(s_1, t_1), \ldots, (s_k, t_k)\}$ \\
  \indent Question: Do there exist vertex-disjoint paths connecting $s_i$ to $t_i$ for each pair $(s_i, t_i) \in \mathcal{R}$?
\end{definition}
\end{mybox}

Node-Disjoint Paths on a Grid can be reduced to Surface Code Routing because the former problem is essentially just the special case of the latter
where there are no dependencies and one time step. The only complication is that not every set of disjoint paths is a valid gate route because of the additional requirement that paths must terminate in a vertical edge at the control qubit
and a horizontal edge at the target. Next, we provide a construction which resolves this issue.

The key idea of this reduction is a pair of constructions called vertex gadgets shown in \cref{fig:vertex_gadgets}.
The \emph{empty vertex gadget} on the left is used to represent vertices that are not included in the set of vertex pairs in the given NDP instance.
The grey filled-in vertices are magic state qubits, which cannot be used to route gates.
The \emph{full vertex gadget}
on the right is used to represent vertices that are not included in the set of vertex pairs in the given NDP instance. The center vertex (in blue)
and the two green vertices diagonal to it are locations where we will map qubits from the circuit we will construct in the reduction.
In particular, we will apply \cnot gates between qubits mapped to the center vertices of different full vertex gadgets and  
between qubits mapped to the diagonal vertices of the same full vertex gadget. For this reason, we refer to the diagonal vertices as the \emph{internal} locations.
As before, the other grey filled-in vertices are magic state qubits.

\begin{figure}[h]
  \begin{subfigure}{0.25\linewidth}
    \includegraphics[width=\linewidth]{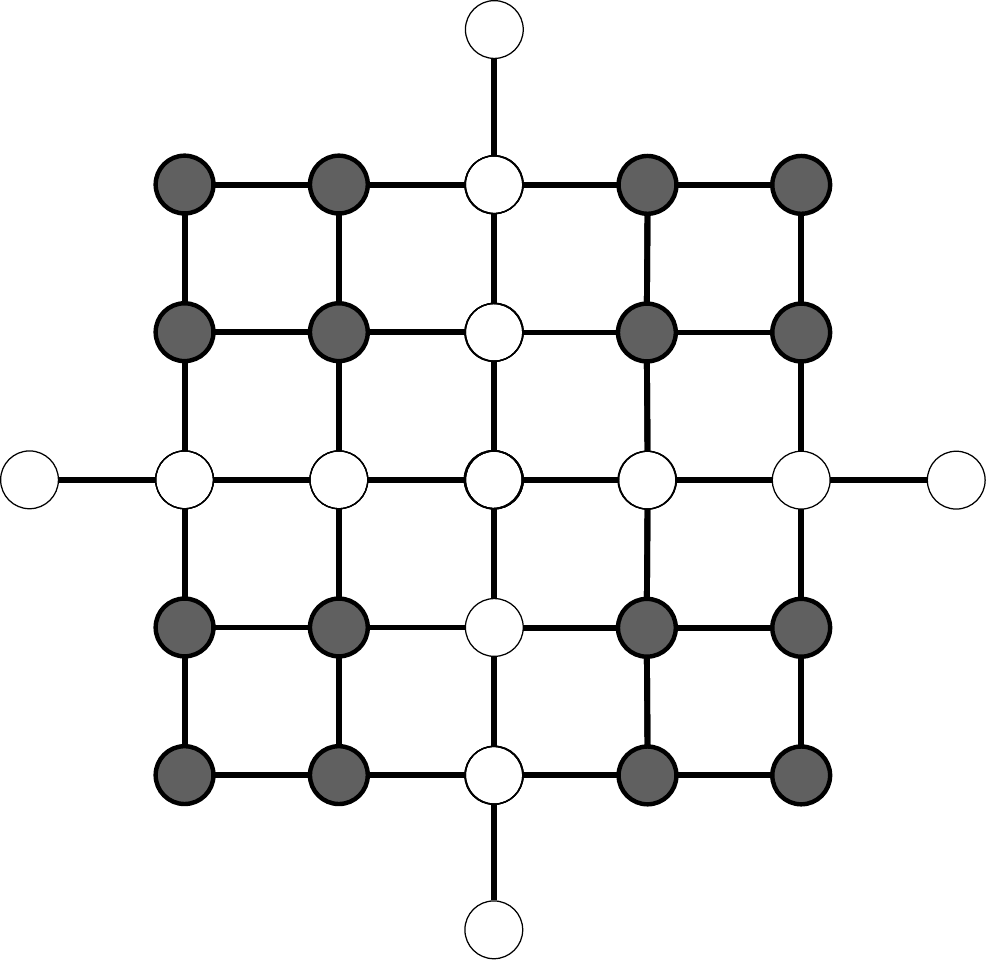}
    \caption{Empty}
    \end{subfigure}
    \hspace{4cm}
  \begin{subfigure}{0.25\linewidth}
    \includegraphics[width=\linewidth]{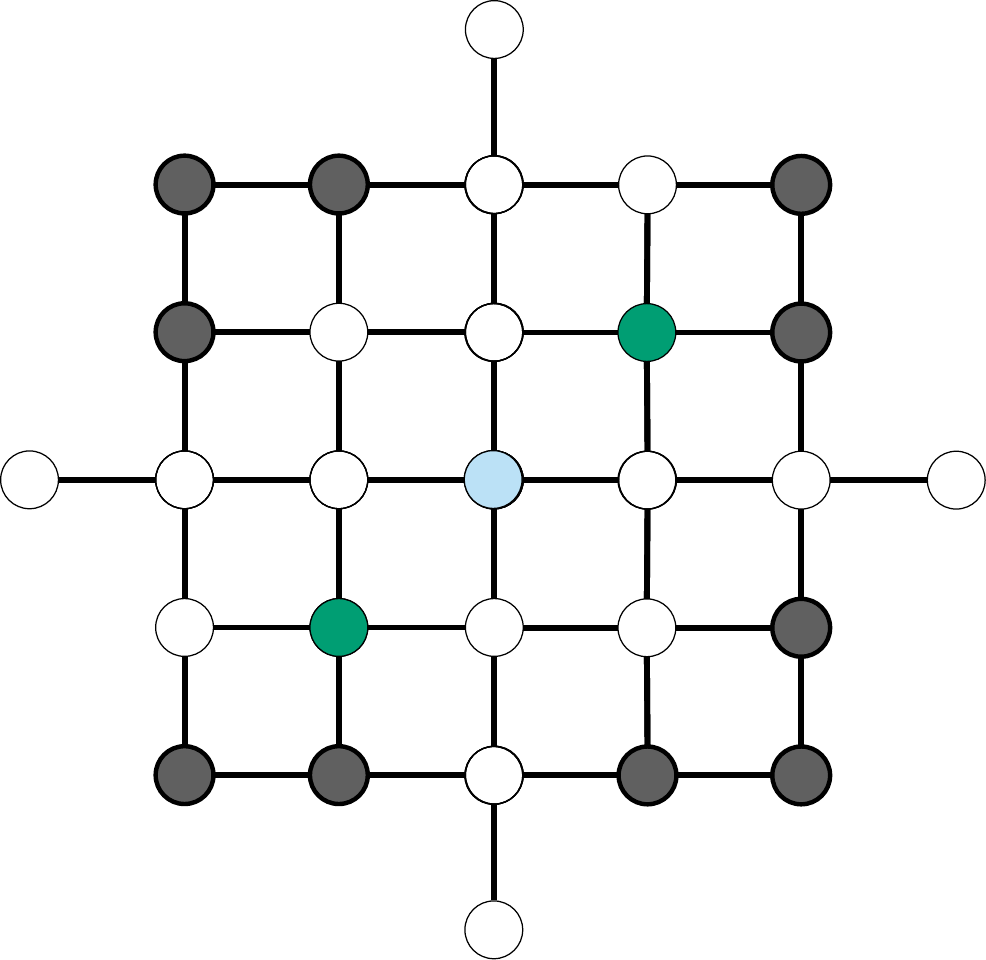}
    \caption{Full}
  \end{subfigure}
  \caption{Vertex gadgets}
  \label{fig:vertex_gadgets}
  \end{figure}

With these gadgets in hand we can describe the reduction. Let $\text{NDP}(G, \mathcal{R})$ be an instance of NDP where $\mathcal{R} = \{(s_1,t_1), \ldots, (s_k, t_k)\}$.
We may assume that each vertex appears in at most one pair in $\mathcal{R}$ as otherwise there is trivially no solution.
We will now construct the corresponding \scr instance $\scr(\arch, \circuit, \map, t_s)$. First, set $t_s=1$. If $G$ is an $n\times m$ grid, then $\arch$,
is chosen to be a $n \times m$ grid of vertex gadgets. In this grid of gadgets, the gadget in grid position $(r,c)$ is the full vertex gadget if vertex $(r,c) \in G$ is
in some pair in $\mathcal{R}$; otherwise, it is the empty vertex gadget.  This process is illustrated in \cref{fig:routing_red} for the $2\times 2$ grid and
$\mathcal{R}$ consisting of a single pair $((1,1), (2,2))$. The solid regions in the figure consist entirely of magic state qubits.
The circuit $\circuit$ includes 
one ``external'' \cnot gate $\cnot ~\text{src}_i ~ \text{tar}_i$ for each pair $(s_i, t_i) \in \mathcal{R}$.  The map $\map$
associates these gates with the proper vertices by mapping $ \text{src}_i$ to the center vertex of the gadget for $s_i$, and mapping  $ \text{tar}_i$
to the center vertex of the gadget for $t_i$. The circuit $\circuit$ also includes one ``internal'' \cnot gate  $\cnot ~\text{tr}_v ~ \text{bl}_v$ for each vertex $v$ that appears in a pair in $\mathcal{R}$,
and $\map$ maps $\text{tr}_v$ to the top-right internal location in the gadget for $v$ and maps $\text{bl}_v$ to the bottom-left internal location in the gadget for $v$.
For the same simple example, this results in three \cnot gates over three disjoint pairs of qubits.

\begin{figure}
  \includegraphics[width=0.65\linewidth]{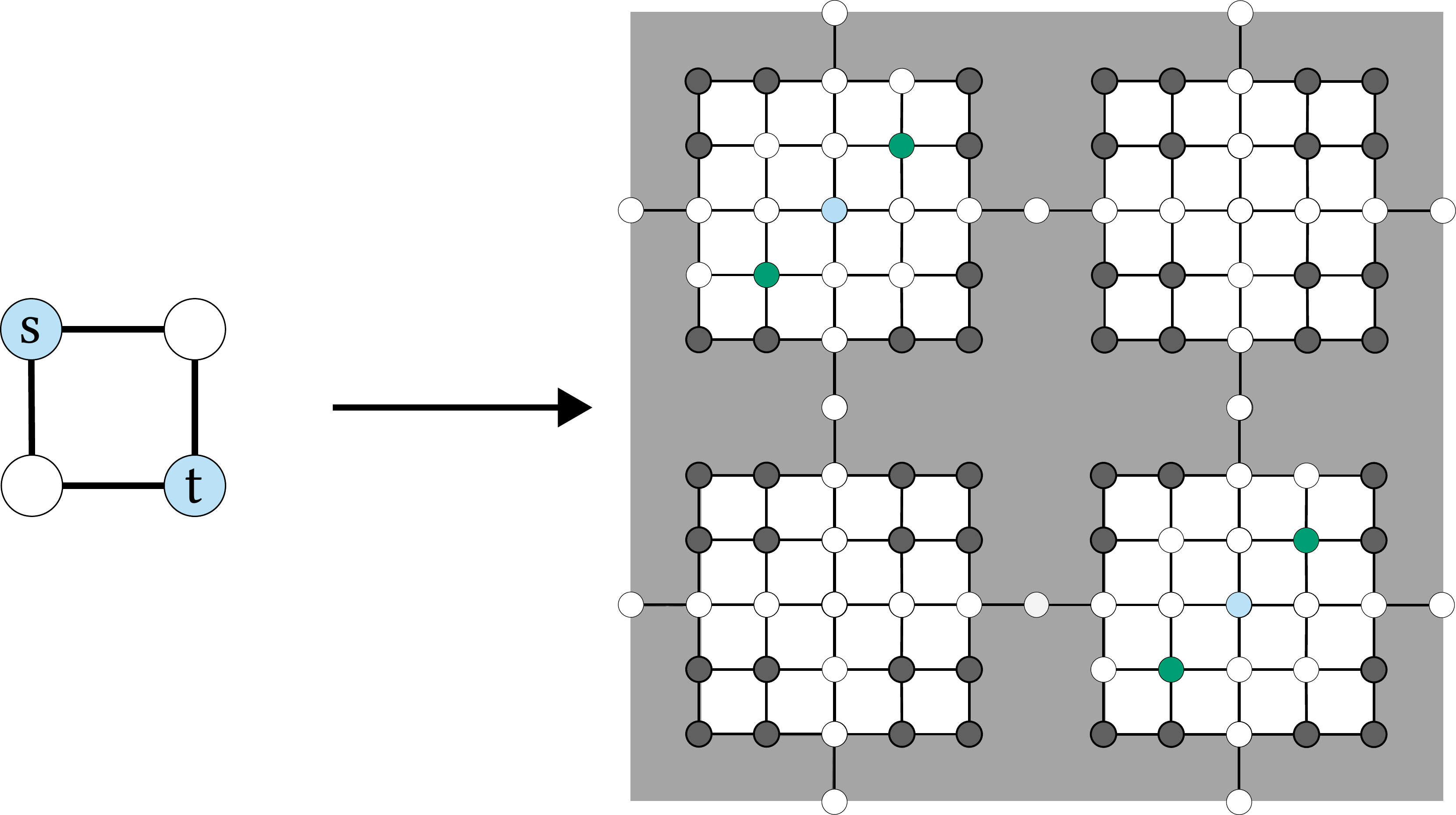}
  \caption{Illustration of the reduction}
  \label{fig:routing_red}
\end{figure}

Next, we show that the answer to $\scr(\arch, \circuit, \map, t_s)$ is ``yes'' if and only if the answer 
to $\text{NDP}(G, \mathcal{R})$ is ``yes.'' The intuition behind the correspondence is shown in \cref{fig:routing_converting_sols},
using an solution to our running example.

\begin{figure}[h]
  \includegraphics[width=0.65\linewidth]{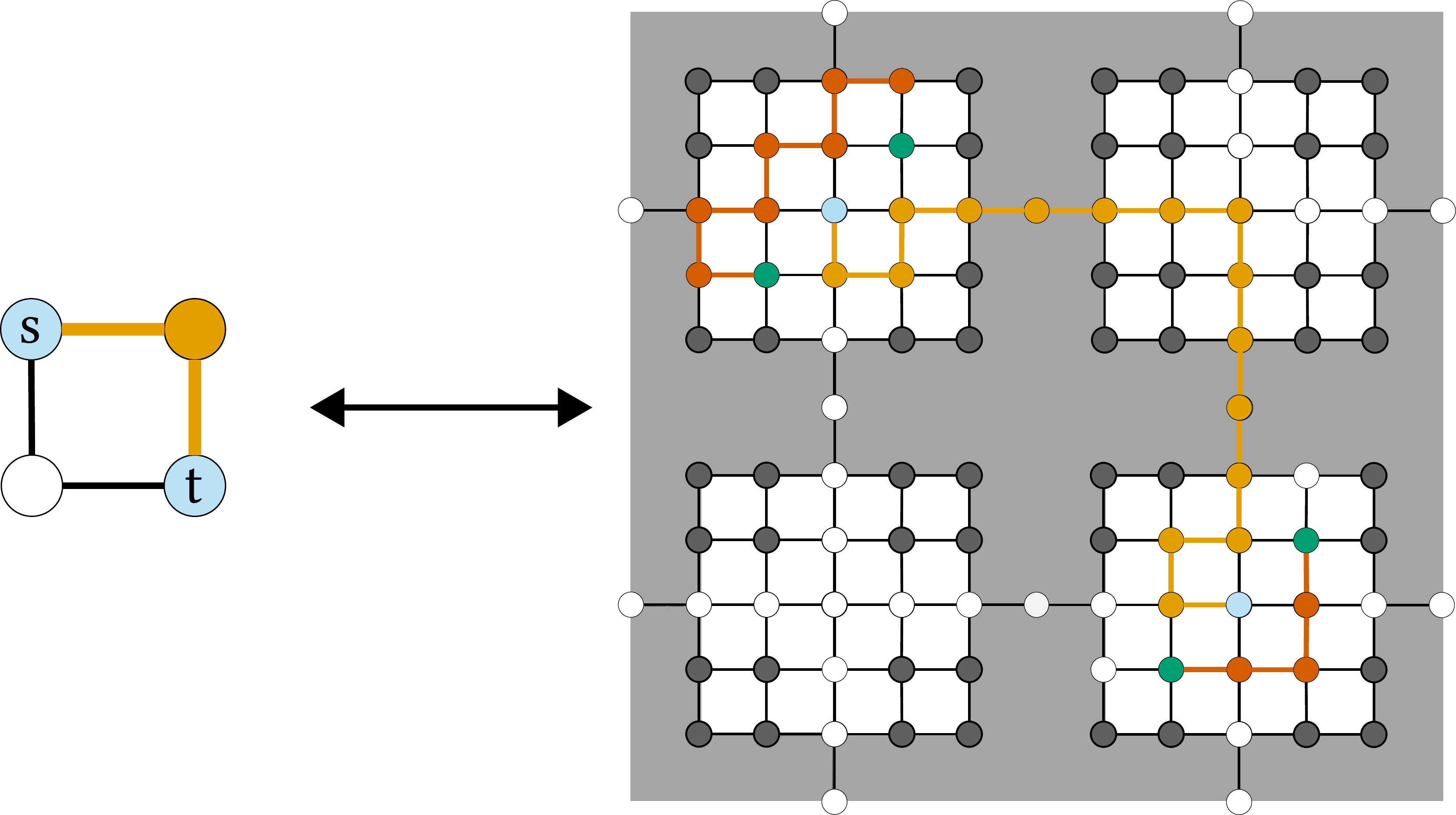}
  \caption{Correspondence between solutions to the two problems}
  \label{fig:routing_converting_sols}
\end{figure}

First suppose there is some set $\mathcal{P}$ of vertex-disjoint paths for routing
the pairs in $\mathcal{R}$. For the external gates,
we can construct a set of vertex-disjoint paths between the gadgets containing their qubits based on $\mathcal{P}$.
In particular, the route for a gate  $\cnot ~\text{src}_i ~ \text{tar}_i$, visits a vertex gadget if and only if the corresponding
vertex is included in the routing for $(s_i, t_i)$ in $\mathcal{P}$, and traverses them in the same order. Then, the final edge of 
the path can be made to have the right orientation by including vertices within the full vertex gadgets at the end points. 
The internal \cnot gates can be simultaneously be routed
using other vertices within the gadget.

Now conversely suppose there is a valid gate route for $\arch$, $\circuit$ and $\map$.
We claim  the routing for two different external \cnot gates cannot use vertices in the same vertex gadget.
Suppose for the sake of contradiction that the routes for two different gates include vertices from the same
vertex gadget. This cannot be an empty vertex gadget because any path that visits an empty vertex gadget
must include the center vertex, but, by definition, valid gate routes are vertex disjoint. Therefore, 
the two paths must both visit the same full vertex gadget corresponding to some vertex $v$. In the case, at least one of them must not
terminate in this gadget, meaning it  enters from and adjacent gadget and exits to another adjacent gadget. However, such a path in conjunction with the route
for the external \cnot gate which does terminate at this gadget make it impossible to route the internal gate 
$\cnot ~\text{tr}_v ~ \text{bl}_v$. Thus, we have reached a contradiction. Now that we have established that at most one gate route
can visit each vertex gadget, we can construct vertex-disjoint paths for the pairs in $\mathcal{R}$
by inverting the process described in the previous direction, where the sequence of vertices used in the path for the pair $(s_i, t_i)$
is exactly the vertices corresponding to the sequence of gadgets used in the route for the gate  $\cnot ~\text{src}_i ~ \text{tar}_i$.

\section{Proof of \cref{thm:enc-correct}}
\label{sec:proof_enc}
Consider an \scc problem consisting of the architecture $\arch$ and circuit $\circuit$. 
First, let $S=(M, t_s, \rspace, \rtime)$ be a solution to the \scc problem where 
$\map$ is a qubit map and $(t_s, \rspace, \rtime)$ is a valid gate route for $\arch, \circuit$, and $\map$.
 Let $\mathcal{I}_{S}$ be the assignment constructed in the following way:
 \begin{itemize}
  \item Set $\vmap(q, v)$ if and only if $M(q) = v$,
  \item Set $\vexec(g, t)$ if and only if $\rtime(g)=t$.
  \item  Set $\vedge(u, v, g, t)$ if and only if $u$ and $v$ are consecutive vertices (with $u$ before $v$) in $\rspace(g)$ and $\rtime(g)=t$. 
 \end{itemize}

We will now show that $\mathcal{I}_{S}$ satisfies $\varphi(A, C, t_s)$  by considering each constraint:

\begin{enumerate}
\item \textsc{map-valid}:  By the fact that $M$ meets our definition of a qubit map, setting the $\vmap(q, v)$ to true if and only if $M(q) = v$ satisfies the $\textsc{map-valid}$ 
constraint. 
\item  \textsc{gates-ordered}: Since $(t_s, \rspace, \rtime)$ is a valid gate route for $\map$, meeting the ``Logical Order'' part of the definition, setting $\vexec(g, t)$ if and only if $\rtime(g) = t$ satisfies the $\textsc{gates-ordered}$ constraint.
\item  \textsc{data-safe}:  A $\vedge$ variable is set to true only when it is indexed by vertices in $\rspace(g)$ for a gate $g$ and a $\vmap$ variable is set to true only when it is indexed by a vertex in the image of $\map$.
The ``Data Preservation'' part of the definition of a valid gate route implies that no vertex can meet both of these criteria and that no magic state vertex can be in an $\rspace$ path. Therefore, $\textsc{data-safe}$ is satisfied.
\item \textsc{disjoint}: Choose a fixed time step $t$. Each vertex $u$ can be included in at most one sequence $\rspace(g)$ where $g$ is such that $\rtime(g)=t$ by the ``Disjoint Paths'' condition on a valid gate route. Since the sequence $\rspace(g)$ is a path, each vertex appears at most once with unique vertices in adjacent positions. Thus, \textsc{disjoint} is satisfied.
\item \textsc{cnot-routed}: The \textsc{cnot-routed} constraint is satisfied because $\rspace$ meets the ``\cnot Routing'' part of the definition of a valid gate route. 
\item \textsc{t-routed}: The \textsc{t-routed} constraint is satisfied because $\rspace$ meets the``\T Routing'' part of the definition of a valid gate route.
\end{enumerate}
Now let $\mathcal{I}$ be an arbitrary model of $\varphi(A, C, t_s)$ and let $S_{\mathcal{I}} = (M, t_s, \rspace, \rtime)$ be derived in the following way:

\begin{itemize}
  \item Set $M(q) = v$ if and only if $\vmap(q, v)$ is set to true in $\mathcal{I}$.
  \item Set $\rtime(g)=t$. if and only if $\vexec(g, t)$ is set to true in $\mathcal{I}$.
  \item Include $u$ and $v$ as consecutive vertices in $\rspace(g)$ if only if and $\vedge(u,v,g,t)$ and $\vexec(g, t)$ are set to true in $\mathcal{I}$ . 
 \end{itemize}

We first argue that $\rspace$ is in fact a well-defined path with the appropriate end points for each gate $g$. Without loss of generality let $g=$ \cnot $q_i$ $q_j$. The \textsc{cnot-start} and \textsc{disjoint} constraints 
mean that  $\rspace(g)$ has exactly one initial vertex which appears as the first index of a path variable but not the second. Moreover, this initial vertex is $M(q_i)$ and the second vertex is a vertical neighbor. Similarly the \textsc{cnot-reach-target} and \textsc{disjoint} constraints 
mean that  $\rspace(g)$ has one final vertex, $M(q_j)$, preceded by a horizontal neighbor. Finally \textsc{cnot-inductive} and \textsc{disjoint} constraints together mean that every other vertex $v$ has exactly one incoming edge (from a path variable where $v$ is the second index) and outgoing edge (from a path variable where $v$ is the first index), ensuring a well-defined linear sequence.
Thus we have established the $\rspace$ function meets the ``\cnot Routing'' and ``\T Routing'' parts of the definition of a valid gate route. 

The last nontrivial step is to show that these paths are vertex-disjoint. Suppose that some vertex $v$ is included in 
$\rspace(g)$ and $\rspace(g')$ where $\rtime(g) = \rtime(g')$. By the \textsc{disjoint} constraint, this can only occur if $\vedge(u,v, g, t)$ and $\vedge(v,w, g', t)$
are both set to true for some vertices $u,w$. But then this means that $v$ is not in the image of $M$ and $v \not\in \msf$ by the \textsc{data-safe} constraint.
Therefore by the \textsc{cnot-inductive} and \textsc{t-inductive} constraints, $g = g'$. 

It is straightforward to see $\map$ is a qubit map because $\mathcal{I}$ satisfies \textsc{map-valid}. Likewise,
it follows that $(t_s, \rspace, \rtime)$ meets the ``Logical Order'' and ``Data Preservation'' conditions
of a valid gate route for $\map$ because $\mathcal{I}$ satisfies \textsc{gates-ordered} and \textsc{data-safe}, respectively.

\section{Additional Results}
\label{sec:plots}
\begin{figure}[h]
  \includegraphics[width=0.285\linewidth]{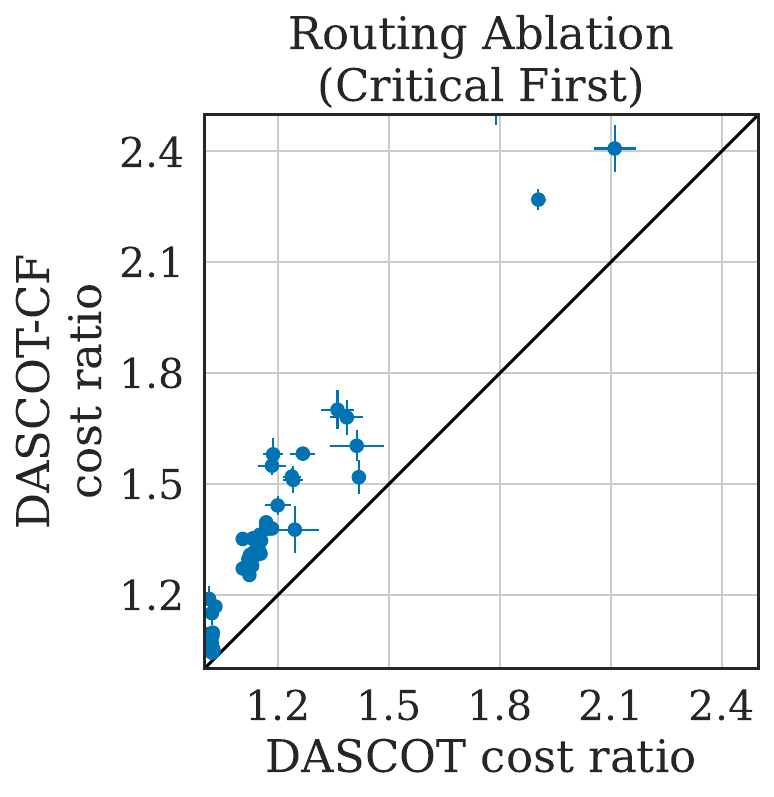}
  \caption{Comparing to critical-first routing}
  \label{fig:rq3}
\end{figure}
\end{document}